%% file: main.tex
\documentclass[10pt,twocolumn]{article}

\usepackage[letterpaper,margin=0.75in]{geometry}
\usepackage[T1]{fontenc}
\usepackage[utf8]{inputenc}
\usepackage{lmodern}
\usepackage{microtype}
\usepackage{amsmath,amssymb,amsfonts,mathtools}
\usepackage{amsthm}
\usepackage{bm}
\usepackage{booktabs}
\usepackage{tabularx}
\usepackage{graphicx}
\usepackage{xcolor}
\usepackage[numbers,sort&compress]{natbib}
\usepackage[hidelinks]{hyperref}
\usepackage[nameinlink,noabbrev]{cleveref}

\input{macros}
\input{theorem-envs}
\input{paper-meta}

\begin{document}

\maketitle

{\small
\noindent\textbf{Abstract.}
\input{sections/00_abstract}

\vspace{0.5\baselineskip}
}

\input{sections/01_introduction}
\input{sections/03_model}
\input{sections/04_mechanism_overview}
\input{sections/05_rho_space_bounds}
\input{sections/06_equilibrium_and_feasibility}
\input{sections/07_extensions_tier2_tier3}
\input{sections/08_validation_and_sensitivity}
\input{sections/09_conclusion_and_scope}
\input{sections/02_related_work}

\section*{Acknowledgments}
\input{sections/11_acknowledgments}

\bibliographystyle{plainnat}
\bibliography{references}

\appendix
\input{appendix/appA_notation_and_scenarios}
\input{appendix/appB_tier1_payoff_decomposition}
\input{appendix/appC_aggregated_and_normalized_bounds}
\input{appendix/appD_cost_adjusted_bounds}
\input{appendix/appE_tier2_tier3_details}
\input{appendix/appF_monte_carlo_and_sensitivity_details}

\end{document}

%% file: macros.tex
\newcommand{\NA}{N_A}
\newcommand{\Nc}{N_c}
\newcommand{\BR}{B_{\mathrm R}}
\newcommand{\BP}{B_{\mathrm P}}
\newcommand{\rrho}{\rho}
\newcommand{\G}{G}
\newcommand{\Y}{Y}
\newcommand{\M}{M}
\newcommand{\U}{U}
\newcommand{\eps}{\varepsilon}
\newcommand{\cstrat}{c}
\newcommand{\ncstrat}{nc}

\newcommand{\dRbar}{\Delta\overline{R}}
\newcommand{\dPbar}{\Delta\overline{P}}
\newcommand{\IR}{\mathrm{IR}}
\newcommand{\IC}{\mathrm{IC}}
\newcommand{\Prb}{\Pr}
\newcommand{\TOne}{\mathrm{T1}}
\newcommand{\TTwo}{\mathrm{T2}}
\newcommand{\TThree}{\mathrm{T3}}

\newcommand{\rwhat}{\widehat{rw}}
\newcommand{\pnhat}{\widehat{pn}}

%% file: theorem-envs.tex
\theoremstyle{definition}

\theoremstyle{plain}
\newtheorem{theorem}{Theorem}

\newtheorem{proposition}{Proposition}
\newtheorem{corollary}{Corollary}

\theoremstyle{remark}
\newtheorem{remark}{Remark}

%% file: paper-meta.tex
\title{A Tunable Incentive Mechanism for Binary Aggregation Without Verification}

\author{
Chien-Chih Chen \qquad Wojciech Golab\\[0.5em]
Department of Electrical and Computer Engineering\\
University of Waterloo\\
Waterloo, Ontario, Canada\\[0.5em]
\texttt{j2255che@uwaterloo.ca} \qquad
\texttt{wgolab@uwaterloo.ca}
}

\date{June 2026}

%% file: sections/00_abstract.tex
Binary aggregation without verifiable ground truth arises when agents' reports must be aggregated without access to gold-standard labels. This paper studies a tunable reward--penalty mechanism for binary aggregation without verification. Agents choose between a conforming strategy, which reports an informative private signal, and a non-conforming strategy, which follows a deterministic prior-informed report rule. For this mechanism, we derive cost-adjusted sufficient conditions for incentive compatibility and individual rationality as bounds on the reward--penalty ratio. The analysis identifies feasible ratio regions, cases in which ratio adjustment restores feasibility, and parameter regimes in which no ratio satisfies both constraints under the modeled construction. We also state a conditional all-conforming Nash equilibrium result within the restricted strategy set. Entropy-based scaling and stake-weighted redistribution are treated as extensions, with stake-weighted redistribution inducing agent-specific incentive constraints. Numerical checks support the closed-form Tier~1 quantities and illustrate threshold sensitivity.

%% file: sections/01_introduction.tex
\section{Introduction}
\label{sec:introduction}

Many multi-agent systems require aggregating binary reports when the
ground-truth label is unavailable or cannot be directly verified. Examples
include peer review, expert forecasting, policy assessment, and crowdsourced
evaluation tasks in which a central evaluator cannot check each response against
a gold-standard answer. In such settings, the mechanism cannot reward agents for
matching the ground truth directly. It must instead use information contained in
the profile of submitted reports. This setting is closely related to information
elicitation and crowdsourcing quality control without direct verification
\cite{jin2020technical,faltings2022game}.

Several nearby lines of work address related but different problems. Majority voting is simple and widely used in crowdsourcing systems \cite{mercier2019majority,xu2018reward,huang2019crowdsourcing}, but it does not by itself calibrate incentives for costly informative reporting. Output agreement, peer prediction, and Bayesian Truth Serum use agreement, report correlations, or belief reports to elicit information \cite{von2004labeling,waggoner2014output,miller2005eliciting,prelec2004bayesian,kong2019information,witkowski2012robust,radanovic2013robust,shnayder2016informed,schoenebeck2020two,liu2023surrogate}. Classical label aggregation methods instead estimate latent labels and worker reliability from noisy reports \cite{dawid1979maximum,whitehill2009whose,zhou2012learning,kim2012bayesian,venanzi2014community,raykar2010learning}. Rather than eliciting beliefs or estimating latent labels, this paper studies incentive calibration for an action-based reward--penalty construction. The mechanism uses only submitted reports, without eliciting agents' beliefs.

Agents choose between a conforming strategy \(\cstrat\), which reports an
informative private signal, and a non-conforming strategy \(\ncstrat\), which
ignores the private signal and follows a deterministic prior-informed report
rule (\Cref{sec:model}). The mechanism assigns rewards and penalties
according to agreement with the final aggregation outcome, not according to
direct verification against the latent ground truth. The central design
parameter is the reward--penalty ratio \(\rrho:=\BR/\BP\).

As a concrete construction, we introduce a tunable three-tier reward--penalty
mechanism, which we call the Tri-Tier Tuning Mechanism (3TTM). Tier~1 uses an
equal-split majority-based reward--penalty rule. Tier~2 applies positive
entropy-based scaling to Tier~1 payoff summands. Tier~3 redistributes selected
base-tier reward and penalty pools according to stake shares, which can make
incentive conditions agent-specific.

The main technical question is when reward and penalty magnitudes can be
calibrated so that conforming behavior is incentive-compatible relative to
\(\ncstrat\) and individually rational. We derive cost-adjusted sufficient
conditions in \(\rrho\)-space. For the base tiers, the normalized reward and
penalty gaps determine whether incentive compatibility (IC) imposes a lower or
an upper bound on \(\rrho\). Thus the feasible region is not always a single
lower threshold: when IC gives an upper bound and individual rationality (IR)
gives a lower bound, the two regions may fail to overlap for this construction.
We also state a conditional all-conforming Nash equilibrium result within the
restricted strategy set \(\{\cstrat,\ncstrat\}\).

\paragraph{Contributions.}
The paper makes three contributions. First, it formulates 3TTM as a tunable action-based reward--penalty construction for binary aggregation without verification. Second, it derives sufficient IC and IR conditions as reward--penalty ratio bounds, including the case in which IC is an upper-bound rather than a lower-bound requirement. These conditions are specific to the proposed construction, not to the problem in general. Third, it gives a conditional all-conforming equilibrium result and identifies parameter regimes in which these sufficient conditions admit no common feasible ratio. Numerical checks are used only to validate the closed-form Tier~1 quantities and to illustrate threshold sensitivity; details are deferred to the appendix. The analysis is restricted to the strategy set \(\{\cstrat,\ncstrat\}\), the majority-honest regime, fixed strategy-level effort costs, and the stated signal model.

\paragraph{Organization.}
\Cref{sec:model} defines the signal environment, the one-shot game over the
restricted strategy set, and the IC and IR requirements. \Cref{sec:mechanism}
instantiates the three-tier reward--penalty construction. \Cref{sec:rho-bounds}
derives the cost-adjusted IC and IR conditions in ratio space and characterizes
the feasible ratio regions, and \Cref{sec:equilibrium-feasibility} connects
these conditions to the conditional all-conforming equilibrium and the
infeasibility result. \Cref{sec:extensions} describes the entropy-scaling and
stake-redistribution extensions, and \Cref{sec:validation-sensitivity} reports
the numerical consistency checks and sensitivity analysis. The full
scenario-wise finite-sum derivations underlying the ratio-space coefficients are
collected in the appendices.

%% file: sections/03_model.tex
\section{Model}
\label{sec:model}

This section formalizes the signal environment, the one-shot game over the
strategy set \(\{\cstrat,\ncstrat\}\), the majority-based aggregation and
reward--penalty rule, and the IC and IR requirements used throughout.

\subsection{Environment}

Let $\NA$ denote the number of active agents. The latent binary label is
$\G \in \{t,f\}$ with prior $\Prb(\G=t)=p$. The mechanism does not observe $\G$.
An agent choosing the conforming strategy $\cstrat$ observes an informative private signal
$X_a \in \{t,f\}$ satisfying
\begin{equation}
  \Prb(X_a=\G)=1-\eps, \qquad \Prb(X_a\neq \G)=\eps,
  \label{eq:signal-model}
\end{equation}
where $\eps\in(0,1/2)$. The condition $\eps<1/2$ makes the private signal
informative: a conforming report is more likely to match the latent label than
not. Signals are conditionally independent given $\G$.

Agents choose from the restricted strategy set $S_a=\{\cstrat,\ncstrat\}$.
The conforming strategy reports $y_a=X_a$. The non-conforming strategy ignores the
private signal and follows the deterministic prior-informed rule
\begin{equation}
  V_{\ncstrat}(p)=
  \begin{cases}
    t, & p>1/2,\\
    f, & p\leq 1/2.
  \end{cases}
  \label{eq:nc-rule}
\end{equation}
All agents choosing $\ncstrat$ submit the same report $\U := V_{\ncstrat}(p)$.

Let $u$ denote the number of agents choosing $\ncstrat$, and let
$\Nc:=\NA-u$ denote the number choosing $\cstrat$. We assume the majority-honest regime
\begin{equation}
  u \leq \left\lfloor (\NA-1)/2 \right\rfloor.
  \label{eq:majority-honest}
\end{equation}

\subsection{Binary Aggregation Game}
\label{sec:binary-aggregation-game}

We model each binary-label task as a one-shot normal-form game. The set of
players is $N=\{1,\ldots,\NA\}$. Each player $a\in N$ chooses a strategy
$s_a\in S_a=\{\cstrat,\ncstrat\}$, where $\cstrat$ denotes conforming behavior
and $\ncstrat$ denotes non-conforming behavior. A strategy profile is denoted by
$s=(s_a,s_{-a})$.

Given a strategy profile and a realization of the latent label and private
signals, each agent submits a binary report $y_a\in\{t,f\}$. If
$s_a=\cstrat$, then $y_a=X_a$. If $s_a=\ncstrat$, then
$y_a=V_{\ncstrat}(p)$. Let $y=(y_1,\ldots,y_{\NA})$ denote the full report
profile.

The aggregation outcome $\Y\in\{t,f,\tau\}$ is computed from the full report
profile $y$. Payoffs are induced by the mechanism's reward--penalty rule and are
evaluated in expectation over the latent label, private signals, induced reports,
and the resulting aggregation outcome. For a tier $T \in \{\mathrm{T1},\mathrm{T2},\mathrm{T3}\}$,
we write $U_a^T(s_a,s_{-a})$ for agent $a$'s ex-ante expected payoff under
tier $T$, before subtracting the fixed strategy cost $C_a$. This notation is
used in \Cref{sec:equilibrium-feasibility} to compare unilateral deviations from
the all-conforming profile.

\subsection{Aggregation Outcome and Payoff Rule}
\label{sec:aggregation-outcome-payoff-rule}

The final aggregation outcome $\Y\in\{t,f,\tau\}$ is determined by majority vote
over all submitted reports, where $\tau$ denotes an unresolved tie. If
$\Y\in\{t,f\}$, reports matching $\Y$ are on the reward side and reports
differing from $\Y$ are on the penalty side. If $\Y=\tau$, no reward or penalty
is assigned.

For analytical decomposition, we also use the majority state $\M$ among agents
choosing $\cstrat$. This quantity is not the final aggregation outcome. It is
used only to group report profiles into payoff-distinct scenarios together with
the latent label $\G$ and the common non-conforming report $\U$. The full
scenario catalog and finite-sum payoff decomposition are given in
\Cref{app:scenarios,app:tier1-payoff-decomposition}.

For the base mechanism, agents on the reward side share a reward pool of size
$\BR$ and agents on the penalty side share a penalty pool of size $\BP$. The
reward--penalty ratio
\begin{equation}
  \rrho := \frac{\BR}{\BP}
  \label{eq:rho-definition-model}
\end{equation}
is the main tuning parameter analyzed in \Cref{sec:rho-bounds}. The penalty
scale $\BP$ determines the absolute payoff scale relative to fixed strategy
costs, while $\rrho$ determines the relative strength of reward and penalty
incentives.

\subsection{Utilities and Incentive Requirements}
\label{sec:utilities-incentives}

Let $R_a(s)$ and $P_a(s)$ denote the reward and penalty induced by the mechanism
when agent $a$ chooses strategy $s\in S_a$. Let $C_a(s)$ denote the fixed
strategy cost, and let $G_a(s)$ denote any external gain. The ex-ante expected
utility is
\begin{equation}
  \mathbb{E}[U_a(s)]
  =
  \mathbb{E}[R_a(s)]
  -
  \mathbb{E}[P_a(s)]
  -
  C_a(s)
  +
  G_a(s).
  \label{eq:expected-utility}
\end{equation}
In the baseline analysis, external gains are set to zero. Fixed effort costs are
kept in the cost-adjusted incentive conditions.

Within the restricted strategy set $S_a=\{\cstrat,\ncstrat\}$, the incentive
compatibility requirement compares conforming behavior against the
non-conforming strategy:
\begin{equation}
  \mathbb{E}[U_a(\cstrat)]
  \geq
  \mathbb{E}[U_a(\ncstrat)].
  \label{eq:ic-model}
\end{equation}
The individual rationality requirement for conforming participation is
\begin{equation}
  \mathbb{E}[U_a(\cstrat)] \geq 0.
  \label{eq:ir-model}
\end{equation}
These are ex-ante conditions, evaluated in expectation over the latent label,
private signals, induced reports, and aggregation outcomes. In particular, Equation~\eqref{eq:ic-model} compares the conforming strategy against the non-conforming
strategy within $S_a$; it is not a dominant-strategy or ex-post guarantee.

\subsection{Scope of the Model}
\label{sec:model-scope}

The model is deliberately restricted. It compares the conforming strategy with a
deterministic prior-informed non-conforming strategy. It does not model arbitrary
Byzantine behavior, collusion, belief reports, repeated interaction, or external
strategic gains. The latent label $\G$ is used for ex-ante payoff analysis but is
not observed by the mechanism. The all-nonconforming profile corresponds to
$u=\NA$ and is outside the majority-honest regime in \eqref{eq:majority-honest}.
The informative-signal assumption $\eps<1/2$ supports the conforming-strategy
payoff analysis, but it does not by itself provide a payoff-based exclusion of
the all-nonconforming profile.

%% file: sections/04_mechanism_overview.tex
\section{The Reward--Penalty Construction}
\label{sec:mechanism}

We now instantiate the model with a tunable three-tier reward--penalty
construction. The mechanism operates only on the submitted reports and the
resulting aggregation outcome \(\Y\). It observes neither the latent label
\(\G\) nor each agent's chosen strategy. Strategy labels are therefore used only
for the ex-ante analysis.

For a report profile \(y=(y_1,\ldots,y_{\NA})\) with strict aggregation outcome
\(\Y(y)\in\{t,f\}\), define
\begin{equation}
    W(y):=\{a:y_a=\Y(y)\},
    \qquad
    L(y):=\{a:y_a\ne \Y(y)\}.
    \label{eq:reward-penalty-sides}
\end{equation}
Agents in \(W(y)\) are on the reward-receiving side and agents in \(L(y)\) are
on the penalty-bearing side. If \(\Y(y)=\tau\), no reward or penalty is assigned.
Thus payoffs are based on agreement with \(\Y\), not agreement with the
unobserved latent label \(\G\).

\begin{table}[t]
\centering
\small
\caption{Mechanism components and their analytical roles.}
\label{tab:mechanism-components}
\begin{tabularx}{\columnwidth}{@{}l X X@{}}
\toprule
Tier & Mechanism role & Analytical effect \\
\midrule
Tier~1 & Equal-split reward--penalty rule & Common IC/IR bounds \\
Tier~2 & Positive entropy-based scaling & Reweighted base-tier coefficients \\
Tier~3 & Stake-weighted redistribution & Agent-specific IC/IR constraints \\
\bottomrule
\end{tabularx}
\end{table}

\paragraph{Tier~1.}
Tier~1 is the equal-split baseline. When \(\Y(y)\in\{t,f\}\), the base reward
\(\BR\) is split equally among agents in \(W(y)\), and the base penalty \(\BP\)
is split equally among agents in \(L(y)\). The realized Tier~1 payoff of agent
\(a\) is
\begin{equation}
F_a^{\TOne}(y)=
\begin{cases}
\BR/|W(y)|, & a\in W(y),\\
-\BP/|L(y)|, & a\in L(y)\text{ and }|L(y)|>0,\\
0, & \text{otherwise}.
\end{cases}
\label{eq:tier1-realized-payoff}
\end{equation}
The appendix decomposes the expectation of \eqref{eq:tier1-realized-payoff} into
six payoff-distinct scenarios and produces the scale-free reward and penalty
coefficients used in the ratio-space bounds.

\paragraph{Tier~2 and Tier~3.}
Tier~2 keeps the Tier~1 aggregation rule, scenario partition, and payoff-side
assignment, but positively scales Tier~1 reward and penalty summands using an
entropy-based decisiveness factor, extending the entropy-based incentive idea of
\cite{chen2022implementation} to the present ratio-space framework. Tier~3 keeps the
selected aggregation outcome fixed but redistributes selected base-tier reward
and penalty pools according to stake shares. Hence Tier~2 changes the numerical
coefficients in the same strategy-level ratio-space framework, while Tier~3
generally induces agent-specific IC and IR constraints. The algebraic
definitions of these extensions are given in
\Cref{app:tier2-entropy-scaling,app:tier3-stake-redistribution}.

%% file: sections/05_rho_space_bounds.tex
\section{Reward--Penalty Ratio Bounds}
\label{sec:rho-bounds}
This section gives the ratio-space incentive conditions used by the mechanism.
For the proposed reward--penalty construction, we ask when the reward and
penalty magnitudes can be calibrated so that the conforming strategy is both
incentive-compatible relative to the non-conforming strategy and individually
rational.

Let
\begin{equation}
    \rrho := \frac{\BR}{\BP}, \qquad \BP>0,
    \label{eq:rho-definition}
\end{equation}
where \(\BR\) is the base reward and \(\BP\) is the base penalty.  The
penalty scale \(\BP\) fixes the absolute payoff scale relative to effort
costs, while \(\rrho\) controls the relative strength of rewards and
penalties.  Thus, after substituting \(\BR=\rrho\BP\), effort costs appear
through the normalized terms \(c_{\cstrat}/\BP\) and
\((c_{\cstrat}-c_{\ncstrat})/\BP\).

The payoff quantities used below are induced by the scenario-wise
payoff decomposition of the mechanism.  The full finite-sum definitions
are deferred to \eqref{eq:reward-gap} and \eqref{eq:penalty-gap}. This section states the normalized quantities and the resulting
ratio-space constraints.

\subsection{Base-Tier Normalized Quantities}
\label{subsec:base-tier-quantities}

We first consider a base tier \(B\in\{\TOne,\TTwo\}\).  Tier~1 is the
majority-based equal-split payoff rule.  Tier~2 keeps the same
aggregation rule and payoff assignment rule, but reweights Tier~1
scenario-index summands by positive entropy-based scaling factors.  In
both cases, strategy-level reward and penalty quantities can be
normalized by the number of agents choosing the corresponding strategy.

Let \(\rwhat^B_{\chi}\) and \(\pnhat^B_{\chi}\) denote the scale-free
aggregated reward and penalty coefficients for strategy
\(\chi\in\{\cstrat,\ncstrat\}\) under base tier \(B\).  These quantities are
obtained after factoring out \(\BR\) from reward terms and \(\BP\) from
penalty terms.  For \(\Nc>0\), the per-agent coefficients for a conforming
agent are
\begin{equation}
    \bar r^B_{\cstrat}:=\frac{\rwhat^B_{\cstrat}}{\Nc},
    \qquad
    \bar p^B_{\cstrat}:=\frac{\pnhat^B_{\cstrat}}{\Nc}.
    \label{eq:normalized-c-coefficients}
\end{equation}
For \(u>0\), the corresponding per-agent coefficients for a
non-conforming agent are
\begin{equation}
    \bar r^B_{\ncstrat}:=\frac{\rwhat^B_{\ncstrat}}{u},
    \qquad
    \bar p^B_{\ncstrat}:=\frac{\pnhat^B_{\ncstrat}}{u}.
    \label{eq:normalized-nc-coefficients}
\end{equation}
When \(u=0\), there is no non-conforming group to normalize. The
strategy-comparison bound is therefore stated for \(u>0\).  The
all-conforming deviation comparison used for equilibrium is handled
separately in \Cref{sec:equilibrium-feasibility}.

Using these coefficients, the expected payoff before fixed effort costs
can be written as
\begin{equation}
    \BR \cdot \bar r^B_{\chi}-\BP \cdot \bar p^B_{\chi},
    \qquad \chi\in\{\cstrat,\ncstrat\}.
    \label{eq:base-tier-payoff-coefficients}
\end{equation}
The key objects for IC are the normalized scale-free reward and penalty
gaps
\begin{equation}
    \dRbar^B := \bar r^B_{\cstrat}-\bar r^B_{\ncstrat}
    =
    \frac{\rwhat^B_{\cstrat}}{\Nc}
    -
    \frac{\rwhat^B_{\ncstrat}}{u},
    \label{eq:reward-gap}
\end{equation}
\begin{equation}
    \dPbar^B := \bar p^B_{\cstrat}-\bar p^B_{\ncstrat}
    =
    \frac{\pnhat^B_{\cstrat}}{\Nc}
    -
    \frac{\pnhat^B_{\ncstrat}}{u}.
    \label{eq:penalty-gap}
\end{equation}
The sign of \(\dRbar^B\) is central: it determines whether the IC
constraint is a lower bound or an upper bound on \(\rrho\).

\subsection{Cost-Adjusted Incentive Compatibility}
\label{subsec:cost-adjusted-ic}

Let \(c_{\cstrat}\) and \(c_{\ncstrat}\) denote the fixed per-agent effort
costs of choosing strategies \(\cstrat\) and \(\ncstrat\), respectively.
For a base tier \(B\in\{\TOne,\TTwo\}\), the cost-adjusted IC condition is
\begin{equation}
    \BR\bar r^B_{\cstrat}
    -\BP\bar p^B_{\cstrat}
    -c_{\cstrat}
    \ge
    \BR\bar r^B_{\ncstrat}
    -\BP\bar p^B_{\ncstrat}
    -c_{\ncstrat}.
    \label{eq:ic-payoff-comparison}
\end{equation}
Substituting \(\BR=\rrho\BP\), dividing by \(\BP>0\), and using
\Cref{eq:reward-gap,eq:penalty-gap} gives
\begin{equation}
    \rrho\,\dRbar^B
    \ge
    \dPbar^B + \frac{c_{\cstrat}-c_{\ncstrat}}{\BP}.
    \label{eq:cost-adjusted-ic}
\end{equation}
This is the ratio-space IC constraint.

\begin{proposition}[Base-tier IC bound]
\label{prop:base-tier-ic-bound}
Fix a base tier \(B\in\{\TOne,\TTwo\}\), with \(u>0\) and \(\BP>0\).  If
\(\dRbar^B\neq 0\), define
\begin{equation}
    \rrho^B_{\IC}
    :=
    \frac{\dPbar^B + (c_{\cstrat}-c_{\ncstrat})/\BP}{\dRbar^B}.
    \label{eq:cost-adjusted-ic-threshold}
\end{equation}
Then the cost-adjusted IC condition is equivalent to
\begin{equation}
    \rrho
    \begin{cases}
    \ge \rrho^B_{\IC}, & \text{if } \dRbar^B>0,\\[3pt]
    \le \rrho^B_{\IC}, & \text{if } \dRbar^B<0.
    \end{cases}
    \label{eq:ic-bound-direction}
\end{equation}
If \(\dRbar^B=0\), the ratio threshold is undefined and the IC condition
reduces to
\begin{equation}
    \dPbar^B + \frac{c_{\cstrat}-c_{\ncstrat}}{\BP} \le 0.
    \label{eq:ic-degenerate-case}
\end{equation}
\end{proposition}

\begin{proof}[Proof sketch]
Equation~\eqref{eq:cost-adjusted-ic} is obtained by writing the
cost-adjusted payoff comparison in scale-free form.  If \(\dRbar^B>0\),
dividing by \(\dRbar^B\) preserves the inequality direction and yields a
lower bound on \(\rrho\).  If \(\dRbar^B<0\), the direction reverses and
the IC condition becomes an upper bound.  If \(\dRbar^B=0\), no ratio
threshold can be defined, and the remaining residual inequality is
\eqref{eq:ic-degenerate-case}.
\end{proof}

The important point is that IC is not always a lower-bound requirement.
When the conforming strategy has a lower normalized reward coefficient
than the non-conforming strategy, increasing \(\rrho\) can make the IC
comparison worse.  In that case, IC imposes an upper bound on the
reward--penalty ratio rather than a lower bound.

\subsection{Cost-Adjusted Individual Rationality}
\label{subsec:cost-adjusted-ir}

The IR condition requires the expected payoff from choosing
\(\cstrat\), after subtracting its effort cost, to be nonnegative:
\begin{equation}
    \BR\bar r^B_{\cstrat}
    -\BP\bar p^B_{\cstrat}
    -c_{\cstrat}
    \ge 0.
    \label{eq:ir-payoff-comparison}
\end{equation}
Substituting \(\BR=\rrho\BP\) and dividing by \(\BP>0\) gives
\begin{equation}
    \rrho\,\bar r^B_{\cstrat}
    \ge
    \bar p^B_{\cstrat}+\frac{c_{\cstrat}}{\BP}.
    \label{eq:cost-adjusted-ir}
\end{equation}

\begin{proposition}[Base-tier IR bound]
\label{prop:base-tier-ir-bound}
Fix a base tier \(B\in\{\TOne,\TTwo\}\), with \(\Nc>0\) and \(\BP>0\).  If
\(\bar r^B_{\cstrat}>0\), define
\begin{equation}
    \rrho^B_{\IR}
    :=
    \frac{\bar p^B_{\cstrat}+c_{\cstrat}/\BP}{\bar r^B_{\cstrat}}
    =
    \frac{\pnhat^B_{\cstrat}/\Nc+c_{\cstrat}/\BP}
         {\rwhat^B_{\cstrat}/\Nc}.
    \label{eq:cost-adjusted-ir-threshold}
\end{equation}
Then the cost-adjusted IR condition is equivalent to
\begin{equation}
    \rrho \ge \rrho^B_{\IR}.
    \label{eq:ir-lower-bound}
\end{equation}
If \(\bar r^B_{\cstrat}=0\), the IR condition must be evaluated directly
from \eqref{eq:ir-payoff-comparison}.
\end{proposition}

\subsection{Feasible Ratio Regions}
\label{subsec:feasible-ratio-regions}

For each base tier \(B\in\{\TOne,\TTwo\}\), the feasible values of
\(\rrho\) are obtained by intersecting the IR region and the IC region.
When \(\dRbar^B>0\), both IC and IR give lower bounds, so the boundary
ratio is
\begin{equation}
    \rrho^{B,*}
    :=
    \max\{\rrho^B_{\IR},\rrho^B_{\IC}\}.
    \label{eq:rho-star-lower-bound-regime}
\end{equation}
In this lower-bound regime, any \(\rrho\ge \rrho^{B,*}\) satisfies both
conditions.

When \(\dRbar^B<0\), IC gives an upper bound while IR gives a lower
bound.  The feasible region is nonempty only if the two regions overlap:
\begin{equation}
    \rrho^B_{\IR} \le \rrho^B_{\IC}.
    \label{eq:upper-lower-overlap}
\end{equation}
If this inequality fails, then no positive reward--penalty ratio can
simultaneously satisfy the base-tier IC and IR constraints for that
parameter tuple.

If \(\dRbar^B=0\), IC imposes no ratio threshold.  The IC condition is
satisfied only when \eqref{eq:ic-degenerate-case} holds; otherwise, no
choice of \(\rrho\) can satisfy IC.  If the residual IC condition holds,
feasibility is determined by the IR condition alone.

\begin{table}[t]
\centering
\small
\caption{Representative feasibility cases for the reward--penalty ratio. Each tuple is \((\NA,u,\eps,p)\).}
\label{tab:representative-feasibility-cases}
\begin{tabular}{@{}llll@{}}
\toprule
Tuple \((\NA,u,\eps,p)\) & Case type & Baseline & Tuned \\
\midrule
\((5,2,0.15,0.50)\)  & Cost-efficient & feasible   & feasible \\
\((10,4,0.10,0.30)\) & Repairable     & infeasible & feasible \\
\((11,5,0.45,0.50)\) & Infeasible     & infeasible & infeasible \\
\bottomrule
\end{tabular}
\end{table}

\Cref{tab:representative-feasibility-cases} illustrates the three
qualitative regimes induced by the ratio-space conditions. In the
cost-efficient case, the reference setting already satisfies the sufficient IC
and IR conditions. In the repairable case, the reference setting fails, but the
derived feasible interval is nonempty, so an appropriate choice of \(\rrho\)
can restore feasibility. In the infeasible case, the IC condition imposes an
upper bound while the IR condition imposes a lower bound, and the two bounds do
not overlap. This last case is a construction-specific empty interval for the
proposed reward--penalty mechanism, not a problem-intrinsic impossibility
theorem.

\begin{remark}[Construction-specific scope]
\label{rem:rho-bounds-scope}
The bounds in this section give construction-specific sufficient
conditions for satisfying IC and IR under the proposed reward--penalty
mechanism.  They identify feasible ranges of \(\rrho\) induced by this
mechanism under the restricted strategy set \(\{\cstrat,\ncstrat\}\), the
majority-honest regime, and the stated signal model.  The broader question of whether the infeasible regions reflect a mechanism-class limitation is left as a future direction.
\end{remark}

\subsection{Tier Instantiations}
\label{subsec:tier-instantiations}

For Tier~1, the coefficients \(\rwhat^{\TOne}_{\chi}\) and
\(\pnhat^{\TOne}_{\chi}\) are obtained by summing the six payoff-distinct
scenario contributions under the majority-based equal-split reward and
penalty rule.  The expanded finite-sum expressions are bookkeeping-heavy
and are therefore collected in
\Cref{app:tier1-payoff-decomposition,app:aggregated-normalized-bounds}.
The main text uses only the normalized coefficients and gaps required by
\Cref{prop:base-tier-ic-bound,prop:base-tier-ir-bound}.

For Tier~2, the same ratio-space framework applies after replacing each
Tier~1 summand by its entropy-scaled counterpart.  Because the scaling
factors are positive, Tier~2 preserves the scenario partition and payoff
orientation of Tier~1, although the numerical values of
\(\dRbar^{\TTwo}\), \(\dPbar^{\TTwo}\), \(\rrho^{\TTwo}_{\IC}\), and
\(\rrho^{\TTwo}_{\IR}\) may differ.  The scaling construction is stated
in \Cref{app:tier2-entropy-scaling}.

Tier~3 is not represented by a common strategy-level threshold of the
form above.  Since stake-weighted redistribution can give different
payoff shares to agents choosing the same strategy, Tier~3 generally
induces agent-specific IC and IR bounds.  We therefore discuss Tier~3
separately in \Cref{sec:extensions} and give its full agent-level
construction in \Cref{app:tier3-stake-redistribution}.

%% file: sections/06_equilibrium_and_feasibility.tex
\section{Equilibrium and Feasibility}
\label{sec:equilibrium-feasibility}

This section states the all-conforming equilibrium implication induced by the
payoff rules. The result is conditional: it applies when the cost-adjusted payoff
difference for a unilateral deviation from \(\cstrat\) to \(\ncstrat\) is
nonnegative. By the majority-honest regime \(u\le\lfloor(\NA-1)/2\rfloor\), the
all-nonconforming profile, which has \(u=\NA\), is outside the regime considered
here.

\paragraph{Deviation gaps.}
Let \(U_a^B(s_a,s_{-a})\) denote agent \(a\)'s ex-ante expected payoff before
fixed strategy costs under base tier \(B\in\{\TOne,\TTwo\}\), with the
aggregation outcome recomputed under the compared profile. Define the
cost-adjusted base-tier deviation gap
\begin{equation}
\begin{aligned}
D_{\cstrat\to\ncstrat}^B
:={}&
\bigl(U_a^B(\cstrat,s_{-a}=\cstrat)-c_{\cstrat}\bigr)\\
&-
\bigl(U_a^B(\ncstrat,s_{-a}=\cstrat)-c_{\ncstrat}\bigr).
\end{aligned}
\label{eq:base-tier-deviation-payoff-difference}
\end{equation}
For Tier~1 and Tier~2 this quantity is common across agents under the
ex-ante symmetric base-tier payoff structure. For Tier~3 with selected base
tier \(B\), the corresponding agent-specific gap is
\begin{equation}
\begin{aligned}
D_{\cstrat\to\ncstrat,a}^{\TThree\mid B}
:={}&
\bigl(U_a^{\TThree\mid B}(\cstrat,s_{-a}=\cstrat)-c_{\cstrat}\bigr)\\
&-
\bigl(U_a^{\TThree\mid B}(\ncstrat,s_{-a}=\cstrat)-c_{\ncstrat}\bigr).
\end{aligned}
\label{eq:tier3-deviation-payoff-difference}
\end{equation}

\begin{theorem}[Conditional all-conforming equilibrium]
\label{thm:conditional-all-conforming-equilibrium}
Assume \(\NA\ge 2\). For a base tier \(B\in\{\TOne,\TTwo\}\), if
\[
    D_{\cstrat\to\ncstrat}^B\ge 0,
\]
then the all-conforming profile is a Nash equilibrium within
\(S_a=\{\cstrat,\ncstrat\}\) under \(B\). For Tier~3 with selected base tier
\(B\), if
\[
    D_{\cstrat\to\ncstrat,a}^{\TThree\mid B}\ge 0
    \quad\text{for all }a\in N,
\]
then the all-conforming profile is a Nash equilibrium within
\(S_a=\{\cstrat,\ncstrat\}\) under \(\TThree\mid B\).
\end{theorem}

\begin{proof}[Proof sketch]
At the all-conforming profile, the only unilateral deviation in the restricted
strategy set is from \(\cstrat\) to \(\ncstrat\). The displayed inequality says
that this deviation is not profitable after fixed strategy costs are subtracted.
For Tier~3 the same argument is applied agent by agent, because
stake-weighted redistribution can make the deviation gap agent-specific.
\end{proof}

The theorem is separate from participation: IR additionally requires
nonnegative cost-adjusted payoff from \(\cstrat\). It also does not imply
dominant-strategy truthfulness, equilibrium uniqueness, or a payoff-based
exclusion of the all-nonconforming profile.

\paragraph{Construction-specific infeasibility.}
The ratio-space bounds in \Cref{sec:rho-bounds} also identify parameter tuples
for which the proposed construction has no feasible reward--penalty ratio. In
the lower-bound regime \(\dRbar^B>0\), both IC and IR impose lower bounds on
\(\rrho\). In the upper-bound regime \(\dRbar^B<0\), IC imposes an upper bound
while IR imposes a lower bound.

\begin{corollary}[Construction-specific infeasibility]
\label{cor:construction-specific-infeasibility}
Fix a base tier \(B\in\{\TOne,\TTwo\}\). Suppose \(\dRbar^B<0\), so that IC
imposes an upper bound, and suppose the thresholds in
\Cref{prop:base-tier-ic-bound,prop:base-tier-ir-bound} are well-defined. If
\[
    \rrho_{\IR}^B>\rrho_{\IC}^B,
\]
then no positive reward--penalty ratio satisfies both the base-tier IC and IR
sufficient conditions for that parameter tuple under the proposed construction.
\end{corollary}

\begin{proof}[Proof sketch]
When \(\dRbar^B<0\), IC gives \(\rrho\le \rrho_{\IC}^B\), while IR gives
\(\rrho\ge \rrho_{\IR}^B\). If the lower bound exceeds the upper bound, the two
regions do not overlap.
\end{proof}

The role of this section is to connect the ratio-space bounds to an equilibrium
statement. \Cref{sec:rho-bounds} identifies values of \(\rrho\) that satisfy
the sufficient IC and IR conditions for the proposed construction. \Cref{thm:conditional-all-conforming-equilibrium} then uses the same cost-adjusted payoff comparison to show that, when a
unilateral deviation from \(\cstrat\) to \(\ncstrat\) is not profitable, the
all-conforming profile is a Nash equilibrium within the restricted strategy set.
Conversely, \Cref{cor:construction-specific-infeasibility} records the case in
which tuning \(\rrho\) cannot repair the construction because the IC upper bound
and IR lower bound form an empty interval.

%% file: sections/07_extensions_tier2_tier3.tex
\section{Entropy Scaling and Stake Redistribution}
\label{sec:extensions}

This section summarizes the analytical roles of the two extensions. The full
definitions are deferred to \Cref{app:tier2-entropy-scaling,app:tier3-stake-redistribution}.

\paragraph{Tier~2.}
Tier~2 is a positive scenario-index scaling of the Tier~1 reward and penalty
summands. It preserves the aggregation rule, payoff-side assignment, scenario
partition, and restricted strategy set. Because every scaling factor is
positive, a reward-side summand remains a reward-side summand and a penalty-side
summand remains a penalty-side summand. Consequently, the same ratio-space IC
and IR structure from \Cref{sec:rho-bounds} applies after replacing the Tier~1
scale-free coefficients by their entropy-scaled Tier~2 counterparts. The
feasible numerical range of \(\rrho\) can nevertheless change because the
summands are reweighted before aggregation.

\paragraph{Tier~3.}
Tier~3 applies stake-weighted redistribution on top of a selected base tier
\(B\in\{\TOne,\TTwo\}\). It keeps the aggregation outcome fixed and redistributes
the selected base-tier reward and penalty pools according to stake shares within
the reward-receiving and penalty-bearing sides. This pool-preservation property
does not preserve individual incentives: two agents choosing the same strategy
can receive different payoff shares when their stakes differ. Therefore Tier~3
generally induces agent-specific IC and IR conditions, and feasibility must be
checked agent by agent.

Both extensions operate within the same restricted strategy set and ratio-space
structure. Tier~2 changes coefficient values within the same base-tier
ratio-space structure; Tier~3 changes the normalization level from common
strategy-level coefficients to agent-specific stake-share coefficients.

%% file: sections/08_validation_and_sensitivity.tex
\section{Numerical Checks and Sensitivity}
\label{sec:validation-sensitivity}

The numerical experiments in this paper play a supporting role. They check the
closed-form Tier~1 quantities and illustrate the behavior of the ratio-space
sufficient conditions; they are not the main theoretical contribution.

For the consistency check, the simulation draws the latent label, generates
conditionally independent reports for agents choosing \(\cstrat\), assigns the
deterministic prior-informed report \(V_{\ncstrat}(p)\) to agents choosing
\(\ncstrat\), and applies the majority aggregation rule to the full report
profile. The resulting per-agent reward and penalty estimates are compared with
the corresponding closed-form Tier~1 quantities used in \Cref{sec:rho-bounds}.
The full simulation settings and summary statistics are reported in
\Cref{app:mc-details}. Across the representative configurations, the
closed-form and simulated quantities agree at the level of the underlying
per-agent reward and penalty components. Larger residuals in
\(\rrho_{\IC}^{\TOne}\) occur when \(\dRbar^{\TOne}\) is close to zero, so the
supporting quantities and the IC-bound direction are more informative than the
ratio alone.

The sensitivity analysis fixes \(p=0.3\) and uses Tier~1 closed-form quantities
to vary the signal error rate and the number of agents choosing \(\ncstrat\).
In the lower-bound regime, the relevant boundary is
\(\rrho^*=\max\{\rrho_{\IR}^{\TOne},\rrho_{\IC}^{\TOne}\}\).

\begin{figure*}[t]
    \centering
    \includegraphics[width=0.82\textwidth]{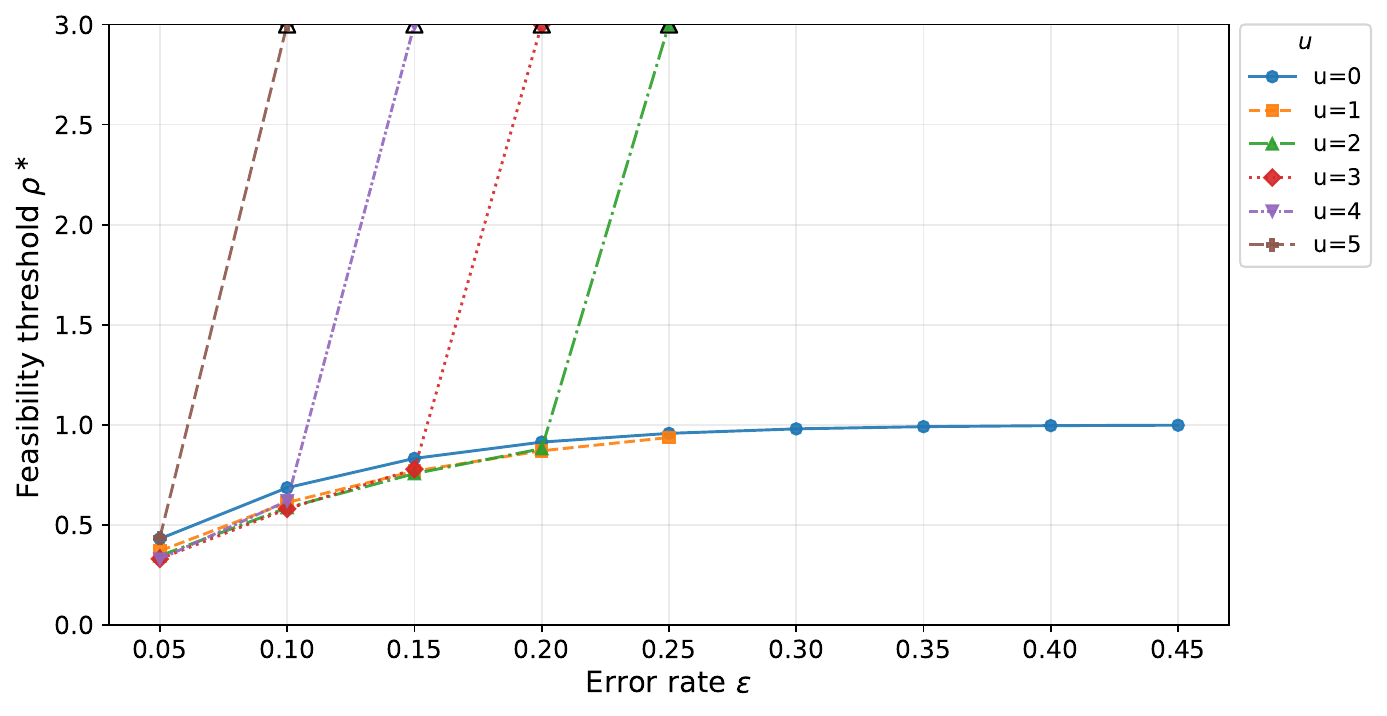}
    \caption{Sensitivity illustration for the Tier~1 lower-bound regime, with
    \(\NA=11\) and \(p=0.3\). The plot shows the feasible threshold
    \(\rrho^*\) as the signal error rate \(\eps\) varies for different values
    of \(u\). Curve termination indicates that the corresponding parameter
    tuple has left the lower-bound regime, rather than missing data.}
    \label{fig:sensitivity-epsilon-u}
\end{figure*}

\Cref{fig:sensitivity-epsilon-u} illustrates the main pattern: higher signal
noise and more non-conforming agents increase the required reward--penalty ratio
in the lower-bound regime. As \(\eps\) and \(u\) increase, the reward-side gap
can approach zero, producing the same small-denominator effect observed in the
consistency check. Additional decomposition plots are collected in
\Cref{app:extra-sensitivity-figures}.

%% file: sections/09_conclusion_and_scope.tex
\section{Conclusion and Scope}
\label{sec:conclusion}

This paper studied a tunable reward--penalty construction for binary aggregation
without ground-truth verification. The analysis derived construction-specific
sufficient IC and IR conditions in reward--penalty ratio space. For the base
tiers, the sign of the normalized reward gap determines whether IC gives a lower
or upper bound on the ratio, while IR gives a lower-bound requirement when the
conforming reward coefficient is positive. These conditions imply a conditional
all-conforming Nash equilibrium within the restricted strategy set and identify
parameter regimes in which the proposed construction admits no feasible ratio.

The scope is deliberately limited to construction-specific sufficient conditions, under the modeled strategy set, majority-honest regime, fixed effort costs, and signal structure. A natural future direction is to ask whether the construction-specific infeasible regions observed here reflect a broader feasibility--informativeness tradeoff for a well-defined class of action-based mechanisms without verification.

%% file: sections/02_related_work.tex
\section{Related Work}
\label{sec:related-work}

This paper is related to label aggregation, information elicitation without verification, and majority-based incentive mechanisms. We focus on how these lines differ from the ratio-space incentive analysis developed here.

\paragraph{Label aggregation.}
Classical label aggregation methods, including Dawid--Skene and later probabilistic models of worker ability, task difficulty, and item characteristics, estimate latent labels and worker reliability from noisy reports~\cite{dawid1979maximum,whitehill2009whose,zhou2012learning,kim2012bayesian,venanzi2014community,raykar2010learning}. These methods address the statistical problem of inferring labels from reports. In contrast, this paper studies incentive calibration for an action-based reward--penalty mechanism when direct verification is unavailable.

\paragraph{Information elicitation without verification.}
Output agreement, peer prediction, Bayesian Truth Serum, and related mechanisms elicit information by using agreement, report correlations, or additional belief reports~\cite{von2004labeling,waggoner2014output,prelec2004bayesian,miller2005eliciting,kong2019information,witkowski2012robust,radanovic2013robust,shnayder2016informed,schoenebeck2020two,liu2023surrogate}. This literature studies incentive properties under assumptions about priors, signal structures, belief reports, and equilibrium selection. The present paper does not introduce a belief-elicitation or proper-scoring mechanism. It instead analyzes construction-specific sufficient IC and IR conditions for a simpler action-based reward--penalty construction.

\paragraph{Majority-based aggregation and incentives.}
Majority voting is widely used because of its simplicity and low computational cost~\cite{mercier2019majority,xu2018reward,huang2019crowdsourcing}. However, majority aggregation alone does not calibrate incentives for costly informative reporting, and agreement-based rewards can admit low-information behavior~\cite{waggoner2014output,faltings2022game,huang2020online,huang2021strategic,huang2020using}. Our mechanism keeps the transparency of majority-based aggregation while making the reward and penalty magnitudes explicit design variables. The contribution is a ratio-space analysis of sufficient IC and IR conditions, rather than a new statistical label estimator or a belief-elicitation mechanism.

%% file: sections/11_acknowledgments.tex
This study was funded by Ripple Labs and the Natural Sciences and Engineering
Research Council of Canada (NSERC). We also thank Dr.~Seyed Majid Zahedi,
Dr.~Kate Larson, and Dr.~Mahesh Tripunitara at the University of Waterloo, and
Dr.~Chen Feng at the University of British Columbia, for their valuable insights
on earlier versions of this work.

%% file: appendix/appA_notation_and_scenarios.tex
\section{Notation and Scenario Catalog}
\label{app:notation-scenarios}
\label{app:scenarios}

This appendix collects the notation and scenario catalog used in the payoff
analysis.  The purpose is only to support the finite-sum derivations in the
subsequent appendices.  The catalog distinguishes the latent label, the
conforming-majority state, and the final aggregation outcome.  In particular,
\(\M\) is an analytical majority state among agents choosing \(\cstrat\), while
\(\Y\) is the final aggregation outcome over the full report profile.

\subsection{Core notation}
\label{app:core-notation}

\begin{table*}[t]
\centering
\small
\renewcommand{\arraystretch}{1.18}
\caption{Core notation used in the scenario-wise payoff decomposition.}
\label{tab:app-core-notation}
\begin{tabular}{p{0.20\textwidth}p{0.74\textwidth}}
\toprule
\textbf{Symbol} & \textbf{Meaning} \\
\midrule
\(\G\in\{t,f\}\) & Latent binary label, where \(t\) denotes true and \(f\) denotes false. The mechanism does not observe \(\G\). \\
\(p\) & Prior probability \(\Pr(\G=t)=p\). \\
\(X_a\in\{t,f\}\) & Private signal observed by agent \(a\) under the conforming strategy \(\cstrat\). \\
\(\eps\in(0,1/2)\) & Signal error probability, with \(\Pr(X_a\neq \G)=\eps\) and \(\Pr(X_a=\G)=1-\eps\). \\
\(\cstrat\) & Conforming strategy: the agent reports the informative private signal. \\
\(\ncstrat\) & Non-conforming strategy: the agent ignores the private signal and follows a deterministic prior-informed report rule. \\
\(V_{\ncstrat}(p)\) & Deterministic report used under \(\ncstrat\): report \(t\) if \(p>1/2\), and report \(f\) otherwise. \\
\(\U\in\{t,f\}\) & Common report submitted by agents choosing \(\ncstrat\), defined as \(\U:=V_{\ncstrat}(p)\). \\
\(\NA\) & Total number of active agents. \\
\(u\) & Number of agents choosing the non-conforming strategy \(\ncstrat\). \\
\(\Nc\) & Number of agents choosing the conforming strategy \(\cstrat\), so \(\Nc:=\NA-u\). \\
\(i,j\) & Numbers of conforming agents on the conforming-majority and conforming-minority sides, respectively, with \(i+j=\Nc\). \\
\(\M\in\{t,f,\tau\}\) & Majority state among conforming agents only. The value \(\tau\) denotes no strict majority among conforming agents. \\
\(\Y\in\{t,f,\tau\}\) & Final aggregation outcome computed from the full report profile. If \(\Y=\tau\), the mechanism assigns no reward or penalty. \\
\(\BR,\BP\) & Base reward pool and base penalty pool. \\
\(\rrho\) & Reward--penalty ratio \(\rrho:=\BR/\BP\). \\
\(C_a(s),G_a(s)\) & Fixed strategy cost and external gain for agent \(a\) under strategy \(s\). In the baseline analysis, external gains are set to zero. \\
\bottomrule
\end{tabular}
\end{table*}

For compactness, the derivations use the following binomial mass function.  For
\(q\in[0,1]\) and integer \(n\), define
\begin{equation}
B_{\Nc}(q,n)
:=
\begin{cases}
\binom{\Nc}{n}q^n(1-q)^{\Nc-n}, & n\in\{0,1,\ldots,\Nc\},\\[4pt]
0, & \text{otherwise}.
\end{cases}
\label{eq:app-binomial-mass}
\end{equation}
Thus \(B_{\Nc}(1-\eps,n)\) is the probability that exactly \(n\) conforming
agents receive correct private signals, and \(B_{\Nc}(\eps,n)\) is the
probability that exactly \(n\) conforming agents receive incorrect private
signals.

\subsection{Scenario catalog}
\label{app:scenario-catalog}

Each logical case is represented by a tuple \((\G,\M,\U)\), where \(\G\) is the
latent label, \(\M\) is the majority state among conforming agents, and \(\U\) is
the common non-conforming report.  Since the labels \(t\) and \(f\) are symmetric
for payoff purposes, the twelve logical cases collapse into the six
payoff-distinct scenarios in \Cref{tab:app-scenario-catalog}.

\begin{table*}[t]
\centering
\small
\renewcommand{\arraystretch}{1.18}
\caption{Reduction from twelve logical cases to six payoff-distinct scenarios. The symbol \(\tau\) denotes no strict majority among conforming agents.}
\label{tab:app-scenario-catalog}
\begin{tabular}{p{0.11\textwidth}p{0.24\textwidth}p{0.56\textwidth}}
\toprule
\textbf{Scenario} & \textbf{Logical cases \((\G,\M,\U)\)} & \textbf{Interpretation} \\
\midrule
\(E_1\) & \((t,t,t)\), \((f,f,f)\) & The conforming-majority state matches the latent label, and the non-conforming report agrees with that majority state. \\
\(E_2\) & \((t,f,f)\), \((f,t,t)\) & The conforming-majority state differs from the latent label, and the non-conforming report agrees with that majority state. \\
\(E_3\) & \((t,t,f)\), \((f,f,t)\) & The conforming-majority state matches the latent label, and the non-conforming report disagrees with that majority state. \\
\(E_4\) & \((t,f,t)\), \((f,t,f)\) & The conforming-majority state differs from the latent label, and the non-conforming report disagrees with that majority state. \\
\(E_5\) & \((t,\tau,t)\), \((f,\tau,f)\) & There is no strict majority among conforming agents, and the non-conforming report matches the latent label. \\
\(E_6\) & \((t,\tau,f)\), \((f,\tau,t)\) & There is no strict majority among conforming agents, and the non-conforming report differs from the latent label. \\
\bottomrule
\end{tabular}
\end{table*}

The scenario catalog is used only for analytical decomposition.  It does not
change the mechanism's information structure: the mechanism observes submitted
reports and computes \(\Y\), but it does not observe the latent label \(\G\) or
whether a report was generated by \(\cstrat\) or \(\ncstrat\).

%% file: appendix/appB_tier1_payoff_decomposition.tex
\section{Tier~1 Scenario-Wise Payoff Decomposition}
\label{app:tier1-payoff-decomposition}
\label{app:tier1-payoffs}

This appendix gives the finite-sum payoff terms for Tier~1.  It supports the
ratio-space derivations in the main text by making explicit how the six
scenarios in \Cref{app:scenarios} contribute to the reward and penalty
quantities.  The quantities below are aggregate strategy-wise quantities before
per-agent normalization.  They already include the relevant scenario probability
factors and binomial masses.

\subsection{Auxiliary notation and conventions}
\label{app:tier1-auxiliary-notation}

Let
\begin{align}
q^+
&:=
\Pr(\G=V_{\ncstrat}(p))
=
\begin{cases}
p, & p>1/2,\\
1-p, & p\le 1/2,
\end{cases}
\label{eq:app-q-plus}\\
q^-&:=1-q^+.
\label{eq:app-q-minus}
\end{align}

Thus \(q^+\) is the probability that the deterministic non-conforming report
matches the latent label, and \(q^-\) is the probability that it does not.
Using the binomial mass function in \eqref{eq:app-binomial-mass}, define
\begin{equation}
B^+(n):=B_{\Nc}(1-\eps,n),
\qquad
B^-(n):=B_{\Nc}(\eps,n).
\label{eq:app-b-plus-minus}
\end{equation}
The term \(B^+(n)\) is used for counts of conforming reports that match the
latent label, while \(B^-(n)\) is used for counts of conforming reports that
differ from the latent label.

For each scenario \(E\) and strategy \(\chi\in\{\cstrat,\ncstrat\}\), let
\(rw^{\mathrm{T1}}_{E,\chi}\) and \(pn^{\mathrm{T1}}_{E,\chi}\) denote the
Tier~1 reward and penalty quantities assigned to the agents choosing strategy
\(\chi\).  Empty sums are interpreted as zero.  When the final aggregation
outcome is a tie, \(\Y=\tau\), Tier~1 assigns no reward or penalty.

The scenario probability factors are
\begin{equation}
q_{E_1}=q_{E_4}=q_{E_5}=q^+,
\qquad
q_{E_2}=q_{E_3}=q_{E_6}=q^-.
\label{eq:app-scenario-probabilities}
\end{equation}

\subsection{Scenarios in which the non-conforming report agrees with the conforming majority}
\label{app:tier1-e1-e2}

In \(E_1\) and \(E_2\), the non-conforming report agrees with the majority
state among conforming agents.  Thus the reward side contains both the
conforming-majority reports and the \(u\) non-conforming reports.  The penalty
side contains only the conforming-minority reports.

For \(E_1\), the conforming-majority reports match the latent label.  Hence the
majority-side counts use \(B^+\), the minority-side counts use \(B^-\), and
\begin{align}
rw^{\mathrm{T1}}_{E_1,\cstrat}
&=
q^+\BR
\sum_{i=\lfloor \Nc/2\rfloor+1}^{\Nc}
B^+(i)\frac{i}{i+u},
\\
rw^{\mathrm{T1}}_{E_1,\ncstrat}
&=
q^+\BR
\sum_{i=\lfloor \Nc/2\rfloor+1}^{\Nc}
B^+(i)\frac{u}{i+u},
\\
pn^{\mathrm{T1}}_{E_1,\cstrat}
&=
q^+\BP
\sum_{j=1}^{\lfloor (\Nc-1)/2\rfloor}
B^-(j),
\\
pn^{\mathrm{T1}}_{E_1,\ncstrat}
&=0.
\label{eq:app-e1-terms}
\end{align}

For \(E_2\), the conforming-majority reports differ from the latent label.
Hence the majority-side counts use \(B^-\), the minority-side counts use
\(B^+\), and
\begin{align}
rw^{\mathrm{T1}}_{E_2,\cstrat}
&=
q^-\BR
\sum_{i=\lfloor \Nc/2\rfloor+1}^{\Nc}
B^-(i)\frac{i}{i+u},
\\
rw^{\mathrm{T1}}_{E_2,\ncstrat}
&=
q^-\BR
\sum_{i=\lfloor \Nc/2\rfloor+1}^{\Nc}
B^-(i)\frac{u}{i+u},
\\
pn^{\mathrm{T1}}_{E_2,\cstrat}
&=
q^-\BP
\sum_{j=1}^{\lfloor (\Nc-1)/2\rfloor}
B^+(j),
\\
pn^{\mathrm{T1}}_{E_2,\ncstrat}
&=0.
\label{eq:app-e2-terms}
\end{align}

\subsection{Scenarios in which the non-conforming report opposes the conforming majority}
\label{app:tier1-e3-e4}

In \(E_3\) and \(E_4\), the non-conforming report disagrees with the majority
state among conforming agents.  The final aggregation outcome depends on
whether the \(u\) non-conforming reports overturn the conforming majority.  For
compactness, define
\begin{equation}
\begin{aligned}
(q_3,H_3,L_3)&:=(q^-,B^+,B^-),\\
(q_4,H_4,L_4)&:=(q^+,B^-,B^+).
\end{aligned}
\label{eq:app-e3-e4-compact-notation}
\end{equation}
Here \(H_s\) is the binomial mass for the conforming-majority side in scenario
\(E_s\), and \(L_s\) is the binomial mass for the conforming-minority side.
Scenario \(E_3\) has a ground-truth-matching conforming majority; scenario
\(E_4\) has a ground-truth-mismatching conforming majority.

For each \(s\in\{3,4\}\), split \(E_s\) into the following subcases:
\(E_{sa}\) is the no-overturn case \(i>j+u\), \(E_{sb}\) is the tie case
\(i=j+u\), and \(E_{sc}\) is the overturn case \(i<j+u\).

\paragraph{No overturn.}
In \(E_{sa}\), the final outcome is the conforming-majority state.  The
conforming-majority side receives the reward, while the conforming-minority
side and the non-conforming agents share the penalty side.  Thus
\begin{align}
rw^{\mathrm{T1}}_{E_{sa},\cstrat}
&=
q_s\BR
\sum_{i=\lfloor(\Nc+u)/2\rfloor+1}^{\Nc}
H_s(i),
\\
rw^{\mathrm{T1}}_{E_{sa},\ncstrat}
&=0,
\\
pn^{\mathrm{T1}}_{E_{sa},\cstrat}
&=
q_s\BP
\sum_{j=1}^{\lceil(\Nc-u)/2\rceil-1}
L_s(j)\frac{j}{j+u},
\\
pn^{\mathrm{T1}}_{E_{sa},\ncstrat}
&=
q_s\BP
\sum_{j=0}^{\lceil(\Nc-u)/2\rceil-1}
L_s(j)\frac{u}{j+u}.
\label{eq:app-esa-terms}
\end{align}

\paragraph{Tie.}
In \(E_{sb}\), the final outcome is unresolved.  For each
\(\chi\in\{\cstrat,\ncstrat\}\),
\begin{equation}
rw^{\mathrm{T1}}_{E_{sb},\chi}=0,
\qquad
pn^{\mathrm{T1}}_{E_{sb},\chi}=0.
\label{eq:app-esb-zero}
\end{equation}

\paragraph{Overturn.}
In \(E_{sc}\), the final outcome is the non-conforming report.  The
conforming-minority side and the non-conforming agents share the reward side,
while the conforming-majority side incurs the penalty.  Thus
\begin{align}
rw^{\mathrm{T1}}_{E_{sc},\cstrat}
&=
q_s\BR
\sum_{j=\lfloor(\Nc-u)/2\rfloor+1}^{\lfloor(\Nc-1)/2\rfloor}
L_s(j)\frac{j}{j+u},
\\
rw^{\mathrm{T1}}_{E_{sc},\ncstrat}
&=
q_s\BR
\sum_{j=\lfloor(\Nc-u)/2\rfloor+1}^{\lfloor(\Nc-1)/2\rfloor}
L_s(j)\frac{u}{j+u},
\\
pn^{\mathrm{T1}}_{E_{sc},\cstrat}
&=
q_s\BP
\sum_{i=\lfloor \Nc/2\rfloor+1}^{\lceil(\Nc+u)/2\rceil-1}
H_s(i),
\\
pn^{\mathrm{T1}}_{E_{sc},\ncstrat}
&=0.
\label{eq:app-esc-terms}
\end{align}

For each \(s\in\{3,4\}\) and \(\chi\in\{\cstrat,\ncstrat\}\), the full scenario
quantities are
\begin{align}
rw^{\mathrm{T1}}_{E_s,\chi}
&=
rw^{\mathrm{T1}}_{E_{sa},\chi}
+
rw^{\mathrm{T1}}_{E_{sb},\chi}
+
rw^{\mathrm{T1}}_{E_{sc},\chi},
\label{eq:app-es-total-reward}
\\
pn^{\mathrm{T1}}_{E_s,\chi}
&=
pn^{\mathrm{T1}}_{E_{sa},\chi}
+
pn^{\mathrm{T1}}_{E_{sb},\chi}
+
pn^{\mathrm{T1}}_{E_{sc},\chi}.
\label{eq:app-es-total-penalty}
\end{align}

\subsection{Tie scenarios among conforming agents}
\label{app:tier1-e5-e6}

Scenarios \(E_5\) and \(E_6\) occur only when the conforming reports are tied.
Let \(i^*=j^*=\Nc/2\).  These scenarios contribute only when \(\Nc\) is even
and \(u>0\); otherwise all four quantities for each tie scenario are zero.

Define
\begin{equation}
\begin{aligned}
(q_5,R_5,P_5)&:=(q^+,B^+(i^*),B^-(j^*)),\\
(q_6,R_6,P_6)&:=(q^-,B^-(i^*),B^+(j^*)).
\end{aligned}
\label{eq:app-e5-e6-compact-notation}
\end{equation}
Here \(R_\ell\) is the binomial mass associated with the reward side in
scenario \(E_\ell\), and \(P_\ell\) is the binomial mass associated with the
penalty side.  For each \(\ell\in\{5,6\}\), when \(\Nc\) is even and \(u>0\),
\begin{align}
rw^{\mathrm{T1}}_{E_\ell,\cstrat}
&=
q_\ell\BR R_\ell\frac{i^*}{i^*+u},
\\
rw^{\mathrm{T1}}_{E_\ell,\ncstrat}
&=
q_\ell\BR R_\ell\frac{u}{i^*+u},
\\
pn^{\mathrm{T1}}_{E_\ell,\cstrat}
&=
q_\ell\BP P_\ell,
\\
pn^{\mathrm{T1}}_{E_\ell,\ncstrat}
&=0.
\label{eq:app-e56-terms}
\end{align}

\subsection{Aggregating Tier~1 quantities}
\label{app:tier1-aggregate-scenario-sums}

The Tier~1 aggregate reward and penalty quantities for each strategy are obtained
by summing over the six payoff-distinct scenarios:
\begin{equation}
\begin{aligned}
rw^{\mathrm{T1}}_{\chi}
&:=
\sum_{k=1}^{6} rw^{\mathrm{T1}}_{E_k,\chi},\\
pn^{\mathrm{T1}}_{\chi}
&:=
\sum_{k=1}^{6} pn^{\mathrm{T1}}_{E_k,\chi},
\qquad
\chi\in\{\cstrat,\ncstrat\}.
\end{aligned}
\label{eq:app-tier1-aggregate-sums}
\end{equation}
These aggregate quantities are the inputs to the normalized reward and penalty
coefficients used in the IC and IR bounds.  The normalization step is handled in
the subsequent appendix; the present appendix only records the scenario-wise
Tier~1 payoff decomposition.

%% file: appendix/appC_aggregated_and_normalized_bounds.tex
\section{Aggregated and Normalized Bounds}
\label{app:aggregated-normalized-bounds}

This appendix converts the scenario-wise Tier~1 payoff decomposition in
\Cref{app:tier1-payoff-decomposition} into the scale-free aggregate and
per-agent normalized coefficients used in the ratio-space bounds of
\Cref{sec:rho-bounds}.  The purpose is bookkeeping: the main text uses the
normalized coefficients, while \Cref{app:tier1-payoff-decomposition} contains
the finite-sum scenario terms from which those coefficients are computed.

Throughout this appendix, empty sums are interpreted as zero.  The
non-conforming normalized quantities are stated for \(u>0\).  When \(u=0\),
there is no non-conforming group to normalize, and the all-conforming deviation
comparison is handled separately in \Cref{sec:equilibrium-feasibility}.

\subsection{Scale-free aggregate coefficients}
\label{app:scale-free-aggregate-coefficients}

For Tier~1, define the scale-free scenario coefficients by factoring out the
base reward from reward terms and the base penalty from penalty terms:
\begin{equation}
\begin{aligned}
\widehat{rw}^{\TOne}_{E_k,\chi}
&:=
\frac{rw^{\TOne}_{E_k,\chi}}{\BR},
&
\widehat{pn}^{\TOne}_{E_k,\chi}
&:=
\frac{pn^{\TOne}_{E_k,\chi}}{\BP},
\end{aligned}
\label{eq:app-scale-free-scenario-coefficients}
\end{equation}
for each scenario \(E_k\), \(k\in\{1,\ldots,6\}\), and strategy
\(\chi\in\{\cstrat,\ncstrat\}\).  The quantities
\(rw^{\TOne}_{E_k,\chi}\) and \(pn^{\TOne}_{E_k,\chi}\) are defined in
\Cref{app:tier1-payoff-decomposition}; each already includes the relevant
scenario probability factor and binomial mass.

The aggregate scale-free Tier~1 coefficients are obtained by summing over the
six payoff-distinct scenarios:
\begin{equation}
\begin{aligned}
\rwhat^{\TOne}_{\chi}
&:=
\sum_{k=1}^{6}
\widehat{rw}^{\TOne}_{E_k,\chi},
&
\pnhat^{\TOne}_{\chi}
&:=
\sum_{k=1}^{6}
\widehat{pn}^{\TOne}_{E_k,\chi}.
\end{aligned}
\label{eq:app-tier1-scale-free-sums}
\end{equation}
Equivalently, the aggregate reward and penalty quantities can be written as
\begin{equation}
\begin{aligned}
rw^{\TOne}_{\chi}
&=
\BR\,\rwhat^{\TOne}_{\chi},
&
pn^{\TOne}_{\chi}
&=
\BP\,\pnhat^{\TOne}_{\chi}.
\end{aligned}
\label{eq:app-tier1-factorization}
\end{equation}
This factorization is the step that separates the mechanism's combinatorial
payoff structure from the tunable magnitudes \(\BR\) and \(\BP\).

\subsection{Per-agent normalization for base tiers}
\label{app:base-tier-normalization}

Let \(B\in\{\TOne,\TTwo\}\) denote a base tier.  For \(\Nc>0\), the per-agent
scale-free coefficients for a conforming agent are
\begin{equation}
\bar r^B_{\cstrat}
:=
\frac{\rwhat^B_{\cstrat}}{\Nc},
\qquad
\bar p^B_{\cstrat}
:=
\frac{\pnhat^B_{\cstrat}}{\Nc}.
\label{eq:app-normalized-c-coefficients}
\end{equation}
For \(u>0\), the corresponding coefficients for a non-conforming agent are
\begin{equation}
\bar r^B_{\ncstrat}
:=
\frac{\rwhat^B_{\ncstrat}}{u},
\qquad
\bar p^B_{\ncstrat}
:=
\frac{\pnhat^B_{\ncstrat}}{u}.
\label{eq:app-normalized-nc-coefficients}
\end{equation}
Thus the expected payoff before fixed effort costs can be written compactly as
\begin{equation}
\BR\,\bar r^B_{\chi}
-
\BP\,\bar p^B_{\chi},
\qquad
\chi\in\{\cstrat,\ncstrat\}.
\label{eq:app-normalized-payoff-before-costs}
\end{equation}

The normalized reward and penalty gaps used in the IC comparison are
\begin{equation}
\begin{aligned}
\dRbar^B
&:=
\bar r^B_{\cstrat}
-
\bar r^B_{\ncstrat}
=
\frac{\rwhat^B_{\cstrat}}{\Nc}
-
\frac{\rwhat^B_{\ncstrat}}{u},
\end{aligned}
\label{eq:app-normalized-reward-gap}
\end{equation}
\begin{equation}
\begin{aligned}
\dPbar^B
&:=
\bar p^B_{\cstrat}
-
\bar p^B_{\ncstrat}
=
\frac{\pnhat^B_{\cstrat}}{\Nc}
-
\frac{\pnhat^B_{\ncstrat}}{u}.
\end{aligned}
\label{eq:app-normalized-penalty-gap}
\end{equation}
The different normalization factors in \eqref{eq:app-normalized-reward-gap} and \eqref{eq:app-normalized-penalty-gap} reflect the fact that IC compares the payoff of a representative conforming agent with that of a representative non-conforming agent.

\subsection{Normalized ratio bounds before effort costs}
\label{app:normalized-ratio-bounds-before-costs}

Before adding fixed effort costs, the normalized IR condition for a conforming
agent under base tier \(B\) is
\begin{equation}
\BR\,\bar r^B_{\cstrat}
-
\BP\,\bar p^B_{\cstrat}
\ge 0.
\label{eq:app-uncosted-ir-condition}
\end{equation}
Substituting \(\BR=\rrho\BP\) and dividing by \(\BP>0\) gives
\begin{equation}
\rrho\,\bar r^B_{\cstrat}
\ge
\bar p^B_{\cstrat}.
\label{eq:app-uncosted-ir-ratio-form}
\end{equation}
If \(\bar r^B_{\cstrat}>0\), define the no-cost IR threshold
\begin{equation}
\rrho^{B,0}_{\IR}
:=
\frac{\bar p^B_{\cstrat}}{\bar r^B_{\cstrat}}.
\label{eq:app-uncosted-ir-threshold}
\end{equation}
Then the no-cost IR condition is \(\rrho\ge \rrho^{B,0}_{\IR}\).  If
\(\bar r^B_{\cstrat}=0\), \eqref{eq:app-uncosted-ir-condition} must be checked
directly.

For \(u>0\), the normalized IC condition before effort costs is
\begin{equation}
\BR\,\dRbar^B
\ge
\BP\,\dPbar^B.
\label{eq:app-uncosted-ic-condition}
\end{equation}
Substituting \(\BR=\rrho\BP\) yields
\begin{equation}
\rrho\,\dRbar^B
\ge
\dPbar^B.
\label{eq:app-uncosted-ic-ratio-form}
\end{equation}
If \(\dRbar^B\neq0\), define the no-cost IC threshold
\begin{equation}
\rrho^{B,0}_{\IC}
:=
\frac{\dPbar^B}{\dRbar^B}.
\label{eq:app-uncosted-ic-threshold}
\end{equation}
The direction of the bound is determined by the sign of \(\dRbar^B\):
\begin{equation}
\rrho
\begin{cases}
\ge \rrho^{B,0}_{\IC}, & \text{if } \dRbar^B>0,\\[3pt]
\le \rrho^{B,0}_{\IC}, & \text{if } \dRbar^B<0.
\end{cases}
\label{eq:app-uncosted-ic-bound-direction}
\end{equation}
If \(\dRbar^B=0\), no ratio threshold is defined, and the IC condition reduces
to the direct check \(\dPbar^B\le0\).

\subsection{Connection to Tier~2}
\label{app:tier2-normalized-coefficients}

Tier~2 uses the same aggregation and normalization pipeline as Tier~1 after
replacing each Tier~1 scenario-index summand by its entropy-scaled counterpart.
For a scenario \(E_k\), admissible count realization \(q\in I_k\), and strategy
\(\chi\in\{\cstrat,\ncstrat\}\), the scaled scale-free summands have the form
\begin{equation}
\begin{aligned}
\widehat{rw}^{\TTwo}_{E_k,\chi}(q)
&:=
\sigma_q^{(k)}\,
\widehat{rw}^{\TOne}_{E_k,\chi}(q),
\\
\widehat{pn}^{\TTwo}_{E_k,\chi}(q)
&:=
\sigma_q^{(k)}\,
\widehat{pn}^{\TOne}_{E_k,\chi}(q).
\end{aligned}
\label{eq:app-tier2-scaled-coefficients}
\end{equation}
with \(\sigma_q^{(k)}>0\).  Summing these scaled terms gives
\(\rwhat^{\TTwo}_{\chi}\) and \(\pnhat^{\TTwo}_{\chi}\), after which the
normalization in \eqref{eq:app-normalized-c-coefficients} and \eqref{eq:app-normalized-nc-coefficients} applies without change.

Because the scaling factors are positive, Tier~2 preserves the scenario
partition and payoff-side orientation of Tier~1.  It can nevertheless change
the numerical values of \(\bar r^B_\chi\), \(\bar p^B_\chi\), \(\dRbar^B\), and
\(\dPbar^B\), and therefore the resulting ratio thresholds.  The full definition
of the entropy-scaling factor is given in \Cref{app:tier2-entropy-scaling}.

\begin{remark}[Why this appendix avoids expanded threshold fractions]
The fully expanded IR and IC ratios are obtained by substituting the finite sums
from \Cref{app:tier1-payoff-decomposition} into the normalized coefficients
above.  Writing those substitutions as single fractions creates expressions
that are difficult to read in a two-column format.  This appendix therefore
keeps the long finite sums at the scenario-coefficient level and states the
ratio bounds using the normalized coefficients used in the main text.
\end{remark}

%% file: appendix/appD_cost_adjusted_bounds.tex
\section{Cost-Adjusted Bound Derivations}
\label{app:cost-adjusted-bound-derivations}
\label{app:cost-normalized-bounds}

This appendix records the algebra that converts the normalized reward and
penalty coefficients in \Cref{app:aggregated-normalized-bounds} into the
cost-adjusted IC and IR bounds used in \Cref{sec:rho-bounds}.  The derivation is
kept at the coefficient level rather than substituting the full scenario-wise
finite sums, because the expanded expressions are lengthy in a two-column
format.

Throughout this appendix, fix \(\BP>0\) and define
\(\rrho=\BR/\BP\).  All feasible regions below are understood to be intersected
with the economically meaningful domain \(\rrho>0\).  For base tiers, let
\(B\in\{\TOne,\TTwo\}\).  The normalized base-tier payoff before costs is
\begin{equation}
    \BR\,\bar r^B_{\chi}
    -
    \BP\,\bar p^B_{\chi},
    \qquad
    \chi\in\{\cstrat,\ncstrat\},
    \label{eq:app-cost-normalized-payoff-before-costs}
\end{equation}
where \(\bar r^B_\chi\) and \(\bar p^B_\chi\) are defined in
\Cref{app:base-tier-normalization}.  Let \(c_{\cstrat}\) and
\(c_{\ncstrat}\) be fixed per-agent effort costs.  These costs are not scaled by
\(\BR\) or \(\BP\), so after substituting \(\BR=\rrho\BP\), they appear only
through the normalized quantities \(c_{\cstrat}/\BP\) and
\((c_{\cstrat}-c_{\ncstrat})/\BP\).

\subsection{Base-tier IC derivation}
\label{app:base-tier-ic-derivation}

For \(u>0\), the cost-adjusted IC comparison between choosing \(\cstrat\) and
choosing \(\ncstrat\) is
\begin{equation}
\begin{aligned}
&\BR\,\bar r^B_{\cstrat}
 -\BP\,\bar p^B_{\cstrat}
 -c_{\cstrat}
\\
&\qquad\ge
 \BR\,\bar r^B_{\ncstrat}
 -\BP\,\bar p^B_{\ncstrat}
 -c_{\ncstrat}.
\end{aligned}
\label{eq:app-cost-ic-payoff-comparison}
\end{equation}
Using the reward and penalty gaps
\(\dRbar^B=\bar r^B_{\cstrat}-\bar r^B_{\ncstrat}\) and
\(\dPbar^B=\bar p^B_{\cstrat}-\bar p^B_{\ncstrat}\), this is equivalent to
\begin{equation}
    \BR\,\dRbar^B
    -
    \BP\,\dPbar^B
    \ge
    c_{\cstrat}-c_{\ncstrat}.
    \label{eq:app-cost-ic-gap-form}
\end{equation}
Substituting \(\BR=\rrho\BP\) and dividing by \(\BP>0\) gives
\begin{equation}
    \rrho\,\dRbar^B
    \ge
    \dPbar^B
    +
    \frac{c_{\cstrat}-c_{\ncstrat}}{\BP}.
    \label{eq:app-cost-ic-ratio-form}
\end{equation}
For compactness, define the IC residual
\begin{equation}
    \Theta^B_{\IC}
    :=
    \dPbar^B
    +
    \frac{c_{\cstrat}-c_{\ncstrat}}{\BP}.
    \label{eq:app-ic-residual}
\end{equation}
Then the IC condition is simply
\begin{equation}
    \rrho\,\dRbar^B \ge \Theta^B_{\IC}.
    \label{eq:app-cost-ic-residual-form}
\end{equation}
If \(\dRbar^B\ne0\), the corresponding threshold is
\begin{equation}
    \rrho^B_{\IC}
    :=
    \frac{\Theta^B_{\IC}}{\dRbar^B}.
    \label{eq:app-cost-ic-threshold-compact}
\end{equation}
Thus the direction of the IC bound is determined by the sign of the normalized
reward gap:
\begin{equation}
\rrho
\begin{cases}
\ge \rrho^B_{\IC}, & \text{if } \dRbar^B>0,\\[3pt]
\le \rrho^B_{\IC}, & \text{if } \dRbar^B<0.
\end{cases}
\label{eq:app-cost-ic-direction}
\end{equation}
If \(\dRbar^B=0\), the ratio term drops out and IC reduces to the direct check
\begin{equation}
    \Theta^B_{\IC} \le 0.
    \label{eq:app-cost-ic-zero-gap-check}
\end{equation}

\subsection{Base-tier IR derivation}
\label{app:base-tier-ir-derivation}

For \(\Nc>0\), the cost-adjusted IR condition for an agent choosing
\(\cstrat\) is
\begin{equation}
    \BR\,\bar r^B_{\cstrat}
    -
    \BP\,\bar p^B_{\cstrat}
    -
    c_{\cstrat}
    \ge 0.
    \label{eq:app-cost-ir-payoff-comparison}
\end{equation}
Substituting \(\BR=\rrho\BP\) and dividing by \(\BP>0\) gives
\begin{equation}
    \rrho\,\bar r^B_{\cstrat}
    \ge
    \bar p^B_{\cstrat}
    +
    \frac{c_{\cstrat}}{\BP}.
    \label{eq:app-cost-ir-ratio-form}
\end{equation}
Define the IR residual
\begin{equation}
    \Theta^B_{\IR}
    :=
    \bar p^B_{\cstrat}
    +
    \frac{c_{\cstrat}}{\BP}.
    \label{eq:app-ir-residual}
\end{equation}
Then IR is
\begin{equation}
    \rrho\,\bar r^B_{\cstrat}
    \ge
    \Theta^B_{\IR}.
    \label{eq:app-cost-ir-residual-form}
\end{equation}
When \(\bar r^B_{\cstrat}>0\), define
\begin{equation}
    \rrho^B_{\IR}
    :=
    \frac{\Theta^B_{\IR}}{\bar r^B_{\cstrat}}.
    \label{eq:app-cost-ir-threshold-compact}
\end{equation}
The IR condition is then the lower-bound requirement
\begin{equation}
    \rrho \ge \rrho^B_{\IR}.
    \label{eq:app-cost-ir-lower-bound}
\end{equation}
If \(\bar r^B_{\cstrat}=0\), no ratio threshold is defined; the condition must
be evaluated directly from \eqref{eq:app-cost-ir-payoff-comparison}.  In the
usual case where \(c_{\cstrat}\ge0\) and \(\bar p^B_{\cstrat}\ge0\), this direct
check fails unless both the conforming penalty coefficient and the conforming
cost are zero.

\subsection{Feasible-ratio cases for base tiers}
\label{app:base-tier-feasible-ratio-cases}

The feasible set for base tier \(B\) is the intersection of the IC and IR 
regions.  When \(\dRbar^B>0\), both IC and IR impose lower bounds, so feasibility
is governed by
\begin{equation}
    \rrho^{B,*}
    :=
    \max\{\rrho^B_{\IR},\rrho^B_{\IC}\}.
    \label{eq:app-rho-star-lower-regime}
\end{equation}
Any positive \(\rrho\ge\rrho^{B,*}\) satisfies both cost-adjusted sufficient
conditions.

When \(\dRbar^B<0\), IC imposes an upper bound while IR imposes a lower bound.
The sufficient conditions have a nonempty feasible interval only if
\begin{equation}
    \rrho^B_{\IR}
    \le
    \rrho^B_{\IC}.
    \label{eq:app-upper-lower-overlap-condition}
\end{equation}
If this overlap condition fails, no value of \(\rrho\) satisfies both the
base-tier IC and IR sufficient conditions for the modeled construction and the
given parameter tuple.

When \(\dRbar^B=0\), IC imposes no ratio threshold.  If
\(\Theta^B_{\IC}\le0\), feasibility is determined by IR alone; if
\(\Theta^B_{\IC}>0\), IC cannot be satisfied by changing \(\rrho\).

\subsection{Agent-specific form for Tier~3}
\label{app:tier3-agent-cost-bound-form}

Tier~3 uses the same cost-adjustment algebra, but after stake-weighted
redistribution the reward and penalty coefficients are agent-specific.  Let
\(\bar r^{\TThree\mid B}_a(\chi)\) and
\(\bar p^{\TThree\mid B}_a(\chi)\) denote agent \(a\)'s scale-free coefficients
under Tier~3 applied to base tier \(B\in\{\TOne,\TTwo\}\).  Define
\begin{equation}
\begin{aligned}
\dRbar^{\TThree\mid B}_a
&:=
\bar r^{\TThree\mid B}_a(\cstrat)
-
\bar r^{\TThree\mid B}_a(\ncstrat),
\\
\dPbar^{\TThree\mid B}_a
&:=
\bar p^{\TThree\mid B}_a(\cstrat)
-
\bar p^{\TThree\mid B}_a(\ncstrat).
\end{aligned}
\label{eq:app-tier3-agent-gaps-cost-appendix}
\end{equation}
Then the agent-specific Tier~3 IC condition is
\begin{equation}
    \rrho\,\dRbar^{\TThree\mid B}_a
    \ge
    \dPbar^{\TThree\mid B}_a
    +
    \frac{c_{\cstrat}-c_{\ncstrat}}{\BP}.
    \label{eq:app-tier3-agent-ic-cost-condition}
\end{equation}
Similarly, the Tier~3 IR condition for agent \(a\) choosing \(\cstrat\) is
\begin{equation}
    \rrho\,\bar r^{\TThree\mid B}_a(\cstrat)
    \ge
    \bar p^{\TThree\mid B}_a(\cstrat)
    +
    \frac{c_{\cstrat}}{\BP}.
    \label{eq:app-tier3-agent-ir-cost-condition}
\end{equation}
The associated thresholds and bound directions are obtained exactly as in the
base-tier derivation, replacing the common strategy-level gaps by the
agent-specific quantities in \eqref{eq:app-tier3-agent-gaps-cost-appendix}.  Thus
Tier~3 feasibility must be checked agent by agent.  The stake-share definitions
that determine these agent-level coefficients are given in
\Cref{app:tier3-stake-redistribution}.

\begin{remark}[Why the appendix uses residual notation]
The expanded expressions for \(\rrho^B_{\IC}\) and \(\rrho^B_{\IR}\) are obtained
by substituting the normalized coefficients from
\Cref{app:aggregated-normalized-bounds}.  The residual notation
\(\Theta^B_{\IC}\) and \(\Theta^B_{\IR}\) keeps the derivation readable in a
two-column layout while preserving the same algebraic content as the fully
expanded cost-adjusted thresholds.
\end{remark}

%% file: appendix/appE_tier2_tier3_details.tex
\section{Tier~2 Entropy Scaling and Tier~3 Stake Redistribution Details}
\label{app:tier2-tier3-details}

This appendix gives the algebraic definitions behind the two extensions
summarized in \Cref{sec:extensions}.  Tier~2 modifies the magnitudes of the
Tier~1 reward and penalty summands through a positive entropy-based scaling
factor.  Tier~3 leaves the selected aggregation outcome fixed, but replaces the
equal split of selected payoff pools with stake-weighted redistribution.

\subsection{Tier~2 entropy-based scaling}
\label{app:tier2-entropy-scaling}

Tier~2 uses the same scenario partition, admissible index sets, aggregation
outcome, and payoff-side assignment as Tier~1.  Its only change is to scale each
scenario-index Tier~1 summand before aggregation.

For a payoff-distinct scenario \(E_k\) and an admissible count realization
\(q\in I_k\), let \(n_t^{(k)}(q)\) and \(n_f^{(k)}(q)\) denote the final counts
of reports equal to \(t\) and \(f\), respectively, after the deterministic
non-conforming reports have been included.  Thus
\begin{equation}
    n_t^{(k)}(q)+n_f^{(k)}(q)=\NA.
    \label{eq:app-tier2-count-total}
\end{equation}
The associated vote-split distribution is
\begin{equation}
    P_q^{(k)}
    :=
    \left(
    \frac{n_t^{(k)}(q)}{\NA},
    \frac{n_f^{(k)}(q)}{\NA}
    \right).
    \label{eq:app-tier2-vote-split}
\end{equation}
Let \(H_2(\cdot)\) denote base-2 Shannon entropy.  The entropy-based
decisiveness score is
\begin{equation}
    d^{(k)}(q)
    :=
    1-H_2\!\left(P_q^{(k)}\right).
    \label{eq:app-tier2-decisiveness}
\end{equation}
For binary vote shares, \(H_2(P_q^{(k)})\in[0,1]\), so
\(d^{(k)}(q)\in[0,1]\).  The score is small when the final vote split is close
to even and large when the split is more decisive.  This entropy-based
disagreement measure follows the entropy-based incentive mechanism of
\cite{chen2022implementation}, where Shannon entropy quantifies the degree of voting
disagreement used to set rewards and penalties.

The Tier~2 scaling factor is
\begin{equation}
    \sigma_q^{(k)}
    :=
    s_\beta\!\left(d^{(k)}(q)\right),
    \qquad
    s_\beta(d)
    :=
    1+\beta\left(d-\frac{1}{2}\right),
    \label{eq:app-tier2-scaling-factor}
\end{equation}
where \(\beta\in(0,2)\).  Hence
\begin{equation}
    1-\frac{\beta}{2}
    \le
    \sigma_q^{(k)}
    \le
    1+\frac{\beta}{2},
    \label{eq:app-tier2-scaling-range}
\end{equation}
and in particular \(\sigma_q^{(k)}>0\).  The default numerical setting used in
the sensitivity experiments is \(\beta=1\), which gives
\(\sigma_q^{(k)}\in[0.5,1.5]\).

Let \(\widehat{rw}^{\TOne}_{E_k,\chi}(q)\) and
\(\widehat{pn}^{\TOne}_{E_k,\chi}(q)\) denote the scale-free Tier~1 reward and
penalty summands for strategy \(\chi\in\{\cstrat,\ncstrat\}\) at realization
\((E_k,q)\).  Tier~2 defines the corresponding scale-free summands by
\begin{equation}
\begin{aligned}
\widehat{rw}^{\TTwo}_{E_k,\chi}(q)
&:=
\sigma_q^{(k)}\,
\widehat{rw}^{\TOne}_{E_k,\chi}(q),
\\
\widehat{pn}^{\TTwo}_{E_k,\chi}(q)
&:=
\sigma_q^{(k)}\,
\widehat{pn}^{\TOne}_{E_k,\chi}(q).
\end{aligned}
\label{eq:app-tier2-scaled-summands-expanded}
\end{equation}
Aggregating over all scenario-index terms gives
\begin{equation}
\begin{aligned}
\rwhat^{\TTwo}_{\chi}
&:=
\sum_{k=1}^{6}\sum_{q\in I_k}
\widehat{rw}^{\TTwo}_{E_k,\chi}(q),
\\
\pnhat^{\TTwo}_{\chi}
&:=
\sum_{k=1}^{6}\sum_{q\in I_k}
\widehat{pn}^{\TTwo}_{E_k,\chi}(q).
\end{aligned}
\label{eq:app-tier2-aggregate-coefficients}
\end{equation}
The actual Tier~2 aggregate quantities are therefore
\begin{equation}
    rw^{\TTwo}_{\chi}=\BR\,\rwhat^{\TTwo}_{\chi},
    \qquad
    pn^{\TTwo}_{\chi}=\BP\,\pnhat^{\TTwo}_{\chi}.
    \label{eq:app-tier2-actual-aggregates}
\end{equation}
After these scaled sums are formed, the same normalization pipeline in
\Cref{app:base-tier-normalization} applies.  Thus Tier~2 changes the numerical
values of the normalized coefficients, but not the ratio-space form of the base
tier IC and IR bounds.

\begin{remark}[Why positivity matters]
The condition \(\sigma_q^{(k)}>0\) ensures that Tier~2 does not reverse the
payoff orientation of any scenario-index summand.  A reward-side term remains a
reward-side term and a penalty-side term remains a penalty-side term.  Tier~2
can still change the feasible ratio region because the summands are reweighted
before aggregation.
\end{remark}

\subsection{Tier~3 stake-weighted redistribution}
\label{app:tier3-stake-redistribution}

Tier~3 can be applied on top of either selected base tier
\(B\in\{\TOne,\TTwo\}\).  The selected base tier determines the final
aggregation outcome and the total reward and penalty pool contributions for
each scenario-index realization.  Tier~3 changes only how these selected pools
are divided among agents.

Fix a scenario-index realization \((E_k,q)\).  If the final outcome is
\(\Y_{kq}\in\{t,f\}\), define the reward-receiving and penalty-bearing sides by
\begin{equation}
\begin{aligned}
\mathcal W_{kq}
&:=
\{a: y_a=\Y_{kq}\},
\\
\mathcal L_{kq}
&:=
\{a: y_a\ne\Y_{kq}\}.
\end{aligned}
\label{eq:app-tier3-payoff-sides}
\end{equation}
If \(\Y_{kq}=\tau\), both sides are treated as empty and no reward or penalty
pool is allocated.

Let \(\Omega_a>0\) denote agent \(a\)'s stake.  For nonempty payoff sides, the
Tier~3 stake shares are
\begin{equation}
\phi_a^{\mathrm{maj}}(k,q)
:=
\begin{cases}
\displaystyle
\frac{\Omega_a}{\sum_{b\in\mathcal W_{kq}}\Omega_b},
& a\in\mathcal W_{kq},\\[8pt]
0, & \text{otherwise},
\end{cases}
\label{eq:app-tier3-majority-share}
\end{equation}
\begin{equation}
\phi_a^{\mathrm{min}}(k,q)
:=
\begin{cases}
\displaystyle
\frac{\Omega_a}{\sum_{b\in\mathcal L_{kq}}\Omega_b},
& a\in\mathcal L_{kq},\\[8pt]
0, & \text{otherwise}.
\end{cases}
\label{eq:app-tier3-minority-share}
\end{equation}
When \(\mathcal W_{kq}\) or \(\mathcal L_{kq}\) is empty, the corresponding
shares are defined to be zero.

Let \(RW^B_{E_k}(q)\) and \(PN^B_{E_k}(q)\) denote the selected base-tier total
reward and penalty pool contributions for realization \((E_k,q)\).  The Tier~3
agent-level contributions are
\begin{equation}
\begin{aligned}
rw^{\TThree\mid B}_{E_k}(q,a)
&:=
\phi_a^{\mathrm{maj}}(k,q)\,
RW^B_{E_k}(q),
\\
pn^{\TThree\mid B}_{E_k}(q,a)
&:=
\phi_a^{\mathrm{min}}(k,q)\,
PN^B_{E_k}(q).
\end{aligned}
\label{eq:app-tier3-agent-contributions}
\end{equation}
Because the stake shares sum to one on each nonempty payoff side, summing over
agents recovers the selected base-tier pools:
\begin{equation}
\begin{aligned}
\sum_{a\in\mathcal W_{kq}}
rw^{\TThree\mid B}_{E_k}(q,a)
&=
RW^B_{E_k}(q),
\\
\sum_{a\in\mathcal L_{kq}}
pn^{\TThree\mid B}_{E_k}(q,a)
&=
PN^B_{E_k}(q).
\end{aligned}
\label{eq:app-tier3-pool-identity}
\end{equation}
This is a pool-preservation identity only.  It does not imply that individual
payoffs, IC gaps, or IR thresholds are inherited from the selected base tier.

For ratio-space analysis, write the selected base-tier pools in scale-free form
as
\begin{equation}
\begin{aligned}
RW^B_{E_k}(q)
&=
\BR\,\widehat{RW}^{B}_{E_k}(q),
\\
PN^B_{E_k}(q)
&=
\BP\,\widehat{PN}^{B}_{E_k}(q).
\end{aligned}
\label{eq:app-tier3-scale-free-pools}
\end{equation}
For agent \(a\) and strategy choice \(\chi\in\{\cstrat,\ncstrat\}\), the
agent-level scale-free coefficients used in \Cref{app:tier3-agent-cost-bound-form}
are obtained by evaluating the stake shares and selected base-tier pools under
the corresponding profile and summing over scenario-index terms:
\begin{equation}
\begin{aligned}
\bar r^{\TThree\mid B}_a(\chi)
&:=
\sum_{k=1}^{6}\sum_{q\in I_k}
\phi_a^{\mathrm{maj}}(k,q;\chi)\,
\widehat{RW}^{B}_{E_k}(q;\chi),
\\
\bar p^{\TThree\mid B}_a(\chi)
&:=
\sum_{k=1}^{6}\sum_{q\in I_k}
\phi_a^{\mathrm{min}}(k,q;\chi)\,
\widehat{PN}^{B}_{E_k}(q;\chi).
\end{aligned}
\label{eq:app-tier3-agent-scale-free-coefficients}
\end{equation}
Here the argument \(\chi\) indicates that the quantities are evaluated when
agent \(a\) chooses \(\chi\), with the comparison profile specified by the IC or
IR condition being evaluated.  These coefficients are agent-specific because
both the payoff side membership and the stake shares can depend on the agent and
on the realized scenario-index term.

\begin{remark}[No common Tier~3 threshold]
Even if the selected base tier has common strategy-level coefficients,
Tier~3 generally does not.  Agents choosing the same strategy can receive
different stake shares, so Tier~3 IC and IR conditions must be checked at the
agent level using the coefficients in \eqref{eq:app-tier3-agent-scale-free-coefficients}.
\end{remark}

%% file: appendix/appF_monte_carlo_and_sensitivity_details.tex
\section{Monte Carlo and Additional Sensitivity Details}
\label{app:mc-sensitivity-details}
\label{app:mc-details}

This appendix reports implementation details for the numerical consistency check
and collects supplementary sensitivity plots for the ratio-space analysis in
\Cref{sec:validation-sensitivity}. The simulation is used only to check the
closed-form Tier~1 quantities against Monte Carlo estimates; it is not the main
theoretical contribution of the paper.

\subsection{Monte Carlo Procedure}
\label{app:mc-procedure}

For each parameter tuple, the simulation uses \(10^5\) independent runs. In each
run, the latent label \(\G\in\{t,f\}\) is drawn according to
\[
    \Prb(\G=t)=p,
    \qquad
    \Prb(\G=f)=1-p.
\]
Agents choosing \(\cstrat\) report conditionally independent private signals:
they report \(\G\) with probability \(1-\eps\) and the opposite label with
probability \(\eps\). Agents choosing \(\ncstrat\) ignore their private signals
and submit the deterministic prior-informed report \(V_{\ncstrat}(p)\) defined
in \Cref{eq:nc-rule}. After all reports are generated, the final aggregation
outcome is determined by majority vote over the full report profile. If the
vote is tied, both rewards and penalties are set to zero in that run.

The simulation estimates the normalized per-agent reward and penalty quantities
for strategies \(\cstrat\) and \(\ncstrat\). These estimates are then converted
into scale-free coefficients by removing the reward and penalty magnitudes, so
that the Monte Carlo estimates can be compared with the closed-form quantities
used in \Cref{sec:rho-bounds}.

\subsection{Parameter Tuples}
\label{app:mc-parameters}

Odd-size committees are used for the main consistency runs. These runs use
\(\NA\in\{5,11\}\), \(u\in\{0,2,\lfloor \NA/2\rfloor\}\), error rates
\(\eps\in\{0.05,0.15,0.30,0.45\}\), and prior values
\(p\in\{0.25,0.30,0.50,0.70,0.75\}\). Even-size committees are included to
check tie handling. These runs use \(\NA\in\{8,10\}\), error rates
\(\eps\in\{0.10,0.30\}\), and values of \(u\) chosen so that
\(\Nc=\NA-u\) is even when tie configurations are intended to have nonzero
probability.

\subsection{Validation Quantities}
\label{app:mc-validation-quantities}

The threshold quantities are computed from the same scale-free objects used in
the closed-form derivation:
\begin{equation}
\begin{aligned}
    \rrho_{\IR}^{\TOne}
    &=
    \frac{\widehat p_{\cstrat}^{\TOne}}
         {\widehat r_{\cstrat}^{\TOne}},
    \\
    \rrho_{\IC}^{\TOne}
    &=
    \frac{\dPbar^{\TOne}}{\dRbar^{\TOne}}.
\end{aligned}
\label{eq:appF-mc-thresholds}
\end{equation}
The sign of \(\dRbar^{\TOne}\) determines whether the IC condition is a lower
or upper bound. When \(\dRbar^{\TOne}\) is close to zero, the ratio
\(\rrho_{\IC}^{\TOne}\) can have a large residual even when the underlying
reward and penalty components agree closely. Such cases are treated as
ratio-sensitive rather than as structural inconsistencies when the supporting
quantities, IC direction, and feasible-region classification remain consistent.

\begin{table}[t]
\centering
\scriptsize
\setlength{\tabcolsep}{3pt}
\caption{Representative Monte Carlo consistency summary for Tier~1. Ratio-sensitive cases are accepted when the supporting quantities, IC bound direction, and feasible-region classification remain consistent.}
\label{tab:mc-validation-summary}
\begin{tabular}{@{}p{0.55\columnwidth}cc@{}}
\toprule
Metric & Odd & Even \\
\midrule
Total tuples evaluated & 100 & 50 \\
Accepted tuples & 100/100 & 50/50 \\
\(\rrho_{\IR}^{\TOne}\) direct pass & 100/100 & 50/50 \\
\(\rrho_{\IC}^{\TOne}\) direct pass & 74/100 & 34/50 \\
\(\rrho_{\IC}^{\TOne}\) ratio-sensitive but supported & 26/100 & 16/50 \\
Feasible / infeasible tuples & 66/34 & 37/13 \\
Feasible classification match & 100/100 & 50/50 \\
IC direction match & 100/100 & 50/50 \\
Max per-agent quantity error & 0.0016 & 0.0013 \\
Max \(\rrho_{\IR}^{\TOne}\) error & 0.0064 & 0.0037 \\
Max \(\rrho_{\IC}^{\TOne}\) error & 2.7748 & 0.5216 \\
\bottomrule
\end{tabular}
\end{table}

\subsection{Additional Sensitivity Figures}
\label{app:extra-sensitivity-figures}

\begin{figure*}[htp]
    \centering
    \includegraphics[width=0.95\textwidth]{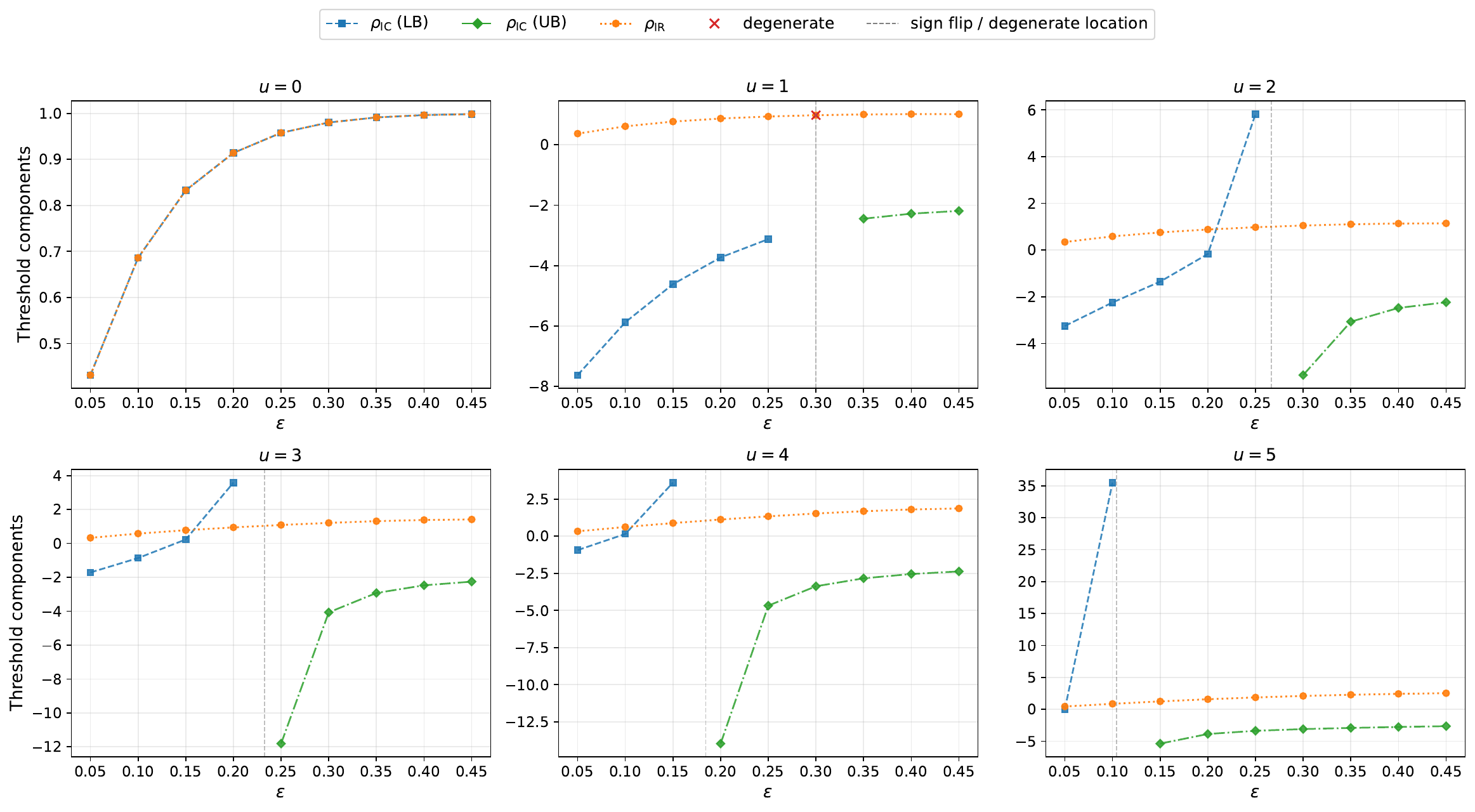}
    \caption{Threshold decomposition across all \(u\) values for the sensitivity path with \(\NA=11\) and \(p=0.3\). The panels show how \(\rrho_{\IR}^{\TOne}\), the lower-bound branch of \(\rrho_{\IC}^{\TOne}\), and the upper-bound branch of \(\rrho_{\IC}^{\TOne}\) vary with \(\eps\). Degenerate markers indicate points where the ratio threshold is not defined because the reward-side gap is approximately zero.}
    \label{fig:sensitivity-decomp-all-u}
\end{figure*}

The main sensitivity plot is shown in \Cref{fig:sensitivity-epsilon-u}. This
subsection provides a supplementary threshold decomposition that supports the
discussion in \Cref{sec:validation-sensitivity}.
\Cref{fig:sensitivity-decomp-all-u} separates the IR threshold from the
lower-bound and upper-bound branches of the IC threshold across all values of
\(u\). Reading across the panels, the upper-bound branch of \(\rrho_{\IC}^{\TOne}\)
appears precisely in the regime where the normalized reward gap \(\dRbar^{\TOne}\)
turns negative; the marked transition in each panel locates this sign change,
which is the same small-denominator regime responsible for the large
\(\rrho_{\IC}^{\TOne}\) residuals reported in \Cref{tab:mc-validation-summary}.
Below this transition, IC imposes a lower bound and feasibility is governed by
\(\max\{\rrho_{\IR}^{\TOne},\rrho_{\IC}^{\TOne}\}\); above it, IC imposes an upper
bound and feasibility requires the overlap condition
\(\rrho_{\IR}^{\TOne}\le\rrho_{\IC}^{\TOne}\).